\documentclass{llncs}
\usepackage{tikz,graphicx,pbox,color}
\usepackage{float}
\usepackage{authblk}
\usepackage{hyperref}
\usepackage{amsmath, amssymb}
\usepackage{enumerate}
\usepackage{graphicx}
\usepackage{microtype}
\usepackage[colorinlistoftodos]{todonotes}
\usepackage[vlined,ruled,noend,linesnumbered]{algorithm2e}
\usepackage{soul}

\usepackage{listings} 
\usepackage{bm}

\newtheorem{openproblem}{Open Problem}
\newtheorem{conj}{Conjecture}

\newcommand{\PMG}{\textbf{PMG}}

\newcommand\ERIC[1]{{\color{black} #1}}

\usepackage[framemethod=TikZ]{mdframed}
\mdfdefinestyle{MyFrame}{%
    linecolor=orange,
    outerlinewidth=1pt,
    roundcorner=20pt,
    innertopmargin=\baselineskip,
    innerbottommargin=\baselineskip,
    innerrightmargin=20pt,
    innerleftmargin=20pt,
    backgroundcolor=yellow!10!white
}

\tikzset{
 treenode/.style = {shape=rectangle, rounded corners,
      draw, align=center},
 root/.style  = {treenode, font=\Large, bottom color=red!30},
 op/.style  = {treenode, font=\ttfamily\normalsize},
 vertex/.style = {circle,draw}
}
\usetikzlibrary{positioning}
\usepgflibrary{shapes}

\author{
  Michel Habib\inst{1}
  \and
  Lalla Mouatadid \inst{2}
  \and
  Eric Sopena \inst{3}
  \and 
  Mengchuan Zou \inst{1}
}
\institute{
  IRIF, UMR 8243 CNRS \& Paris University, Paris, France 
  \and
  Department of Computer Science, University of Toronto, Toronto, ON, Canada 
 \and
 Univ. Bordeaux, CNRS, Bordeaux INP, LaBRI, UMR5800, F-33400 Talence, France
 }

\title{$(\alpha, \beta)$-Modules in Graphs\thanks{This work is supported by the ANR-France Project Hosigra (ANR-17-CE40-0022). Preliminary results of this work were presented in \cite{habib2020approximating}.} }

\begin{document}
\maketitle 
\begin{abstract}
	Modular Decomposition focuses on repeatedly identifying a module $M$ (a collection of vertices that shares \textbf{exactly} the same neighbourhood outside of $M$) and collapsing it into a single vertex. 
    This notion of \emph{exactitude of neighbourhood} is very strict, especially when dealing with real world graphs. 
    
    We study new ways to relax this exactitude condition.  
    However, generalizing modular decomposition is far from obvious. 
    Most of the previous proposals lose algebraic properties of modules and thus most of the nice algorithmic consequences. 
    
    We introduce the notion of an \textbf{\boldmath$(\alpha, \beta)$-module}, a relaxation that allows a bounded number of errors in each node and maintains some of the algebraic structure. 
    It leads to a new combinatorial decomposition with interesting properties. 
    Among the main results in this work, we show that minimal $(\alpha, \beta)$-modules can be computed in polynomial time, and that every graph admits an     $(\alpha, \beta)$-modular decomposition tree, thus generalizing Gallai's Theorem (which corresponds to the case for $\alpha=\beta=0$). 
    Unfortunately we give evidence that computing such a decomposition tree can be difficult. 
    
   \end{abstract}
\section{Introduction}
\label{sec:introduction}
First introduced for undirected graphs by Gallai in \cite{Gal67} to analyze the structure of comparability graphs, modular decomposition has been used and defined in many areas of discrete mathematics, including 2-structures, automata, partial orders, set systems, hypergraphs, clutters, matroids, boolean and submodular functions \cite{EhrenfeuchtHR95,ER90,F91,iwocahypergraph}.
For a survey on modular decomposition, see~\cite{MR84} and for its algorithmic aspects \cite{HabibP10}. 
Since they have been rediscovered in many fields, modules appear under various names in the literature, they have been called intervals, externally related sets, autonomous sets, partitive sets, homogeneous sets, and clans. In most of the above examples the family of modules of a given graph yields a kind of partitive family \cite{BHLM09,Bui-XuanHR12,CHM81}, and therefore leads to a unique modular decomposition tree that can be computed efficiently. 

Roughly speaking, elements of a module $M$ behave exactly the same with respect to elements outside of $M$. 
Thus a module can be contracted to a single element without losing neighbourhood and connectivity information. 
This technique has been used to solve many optimization problems and has led to a number of elegant graph algorithms, see for instance~\cite{Moh85}. 
Other direct applications of modular decomposition appear in areas such as computational protein-protein interaction networks and graph drawing~\cite{GagneurKBC04,PV06}. 
Recently, new applications have appeared in the study of networks in social sciences~\cite{Serafino13}, where a module is considered as a regularity or a community that has to be detected and understood. 

Although it is well known that almost all graphs have no non-trivial modules~\cite{mohring1984almost}, some graphs that arise from real data seem to have many non-trivial modules~\cite{NABTI2017}. 
How can we explain such a phenomenon? 
It could be that the context in which this real data is generated has a clustering structure; but it could also be because we reach some known regularities as predicted by Szemer\'edi's Regularity Lemma~\cite{Szemeredi75}. 
In fact for every $\epsilon > 0$ Szemer\'edi's lemma asserts the existence of an $n_0$ such that all undirected graphs with more than $n_0$ vertices admit an $\epsilon$-regular partition of its vertices. 
Such a partition is a kind of an \emph{approximate} modular decomposition, and linear time algorithms for \emph{exact} modular decomposition are known~\cite{HabibP10}.

\medskip
\noindent
\emph{Our results.}
In this paper we introduce and study a new generalization of modular decomposition by relaxing the strict neighbourhood condition of modules with a tolerance of some errors (missing or extra edges). 
In particular, we define an $(\alpha, \beta)$-module to be a set $M$ 
whose elements behave exactly the same with respect to elements outside of $M$, except that each outside element can have either
at most $\alpha$ missing edges or at most $\beta$ extra edges connecting it to $M$.
In other words, an $(\alpha, \beta)$-module $M$ can be turned into a module by adding at most $\alpha$ edges, or deleting at most $\beta$ edges, at each element outside $M$.
In particular, we recover the standard modular decomposition when $\alpha = \beta = 0$. 

This new combinatorial decomposition is not only theoretically interesting but also can lead to practical applications. We first prove that every graph admits an $(\alpha, \beta)$-modular decomposition tree which is a kind of generalization of Gallai's Theorem. But by no means such a tree is unique and we also give evidence that finding such a tree could be NP-hard. On the algorithmic side we propose a polynomial algorithm to compute a covering of the vertex set by minimal $(\alpha, \beta)$-modules with a bounded overlap, in $O(m \cdot n^{\alpha +\beta +1})$ time.
For the bipartite case, when we restrict $(\alpha, \beta)$-modules on one side of the bipartition, we completely compute all these $(\alpha, \beta)$-modules. In particular, we give an algorithm that computes a covering of the vertices of a bipartite graph in $O(n^{\alpha +\beta}(n + m))$ time, using maximal $(\alpha, \beta)$-modules. This can be of great help for community detection in bipartite graphs.
  
\medskip
\noindent
\emph{Organization of the paper.} 
Section~\ref{sec:background} covers the necessary background on standard modular decomposition, introduces $(\alpha, \beta)$-modules and illustrates various applications of $(\alpha, \beta)$-modular decomposition. 
Sections~\ref{sec:structure} covers structural properties of $(\alpha,\beta)$-modules and the NP-hardness results. 
Section~\ref{sec:algorithms} contains all the algorithmic results, in particular the computation of minimal $(\alpha, \beta)$-modules as well as $(\alpha, \beta)$-primality testing. 
Section~\ref{sec:applications} covers the complete determination of $(\alpha, \beta)$-modules that lay one side of a bipartite graph.
We conclude in Section~\ref{sec:conclusion} with an alternate relaxation of modular decomposition.

\section{Modular Decomposition: A Primer}
\label{sec:background}

Let $G = (V(G), E(G))$ be a graph on $|V(G)| = n$ vertices and $|E(G)| = m$ edges. 
For two adjacent vertices $u, v \in V(G)$, $uv$ denotes the edge in $E(G)$ with endpoints $u$ and $v$. 
All the graphs considered here are simple (no loops, no multiple edges), finite and undirected. 
The complement of a graph $G = (V,E)$ is the graph $\overline{G} = (V(G), \overline{E(G)})$ where $uv \in \overline{E(G)}$ if and only $uv \notin E(G)$. 
We often refer to the sets of vertices and edges of $G$ as $V$ and $E$ respectively, if $G$ is clear from the context. 

For a set of vertices $X \subseteq V$, we denote by $G(X)$ the induced subgraph of $G$ generated by $X$.
The set $N(v)= \{u : uv \in E\}$ is the \emph{neighbourhood} of $v$ and the set $\overline{N(v)}= \{u : uv \notin E\}$ the \emph{non-neighbourhood} of $v$. 
This notation can also be extended to sets of vertices: 
for a set $X\subseteq V$, we let
$$N(X) = \{x \in V\setminus X : \exists y \in X \text{ and } xy \in E(G) \},$$
and 
$$\overline{N}(X) = \{x \in V\setminus X : \forall y \in X,\ xy \notin E(G) \}.$$ 
Note here that $N(X)$ is not the union of the sets $N(x)$ for all $x\in X$, but the set of vertices outside from $X$ that have a neighbour in $X$.

Two vertices $u$ and $v$ are called \emph{false twins} if $N(u) = N(v)$, and \emph{true twins} if $N(u)\cup\{u\} = N(v)\cup\{v\}$. 

A \emph{Moore family} on a set $X$ is a collection of subsets of $X$ that contains $X$ itself and is closed under intersection.

\begin{definition}
\label{def:moduleClassique}
    A \textbf{module}  of a graph $G = (V, E)$ is a set of vertices $M \subseteq V$ that satisfies 
$$\forall x, y \in M,\ N(x) \setminus M = N(y) \setminus M. $$
\end{definition}

In other words, $V \setminus M$ is partitioned into two parts $A, B$ such that there is a complete bipartite subgraph between $M$ and $A$, and no edges between $M$ and $B$. 
Observe that we have $A=N(M)$, and $B=\overline{N}(M)$. 
It is to easy to see that every two vertices within a module are either false twins or true twins. 

A single vertex $\{v : v \in V\}$ is always a module, and so is the set $V$.
Such modules are called \emph{trivial modules}. 
A graph with only trivial modules is called a \emph{prime} graph. 
A module is \emph{maximal} if it is not contained in any other non-trivial module. 

A \emph{modular decomposition tree} of a graph $G$ is a tree $T(G)$ that captures the decomposition of $G$ into modules. 
The leaves of $T(G)$ represent the vertices of $G$, the internal nodes of $T(G)$ capture operations on modules, and are labelled \emph{parallel, series,} or \emph{prime}. 
A parallel node captures the disjoint union of its children, whereas a series node captures the full 
connection of its children (i.e., connect every vertex of its left child to every vertex of its right child).
Parallel and series nodes are often referred to as \emph{complete} nodes. 
Fig.~\ref{fig:first_example} illustrates a graph with its modular decomposition tree. 

\begin{figure}[t]
    \centering
    \begin{tikzpicture}[scale=0.80]
        \coordinate(a) at (-7,0);
        \coordinate(b) at (-6,1);
        \coordinate(c) at (-5.2,0);
        \coordinate(d) at (-6,-1);
        \coordinate(e) at (-4.5,0);
        \coordinate(f) at (-3.7,0.8);
        \coordinate(g) at (-3.7,-0.8);
        \coordinate(h) at (-2.8, 0);
        
        \draw(a)node[left]{$a$} node{$\bullet$};
        \draw(b)node[left]{$b$} node{$\bullet$};
        \draw(c)node[below]{$c$} node{$\bullet$};
        \draw(d)node[left]{$d$} node{$\bullet$};
        \draw(e)node[below]{$e$} node{$\bullet$};
        \draw(f)node[right]{$f$} node{$\bullet$};
        \draw(g)node[right]{$g$} node{$\bullet$};
        \draw(h)node[below]{$h$} node{$\bullet$};
        
        \draw (a) -- (b);
        \draw (a) -- (c);
        \draw (a) -- (d);
        \draw (b) -- (c);
        \draw (c) -- (d);
        \draw (e) -- (b);
        \draw (e) -- (c);
        \draw (e) -- (d);
        \draw (e) -- (f);
        \draw (e) -- (g);
        \draw (f) -- (h);
        \draw (g) -- (h);

        \filldraw[fill=red!30, draw=black, opacity=0.2] (-5.8,0) ellipse (0.7cm and 1.5cm);
        \filldraw[fill=red!30, draw=black, opacity=0.2] (-3.75,0) ellipse (0.5cm and 1.2cm);

        \coordinate(v) at (-1.5,-1);
        \coordinate(x) at (1,-1);
        \coordinate(z) at (3,-1);
        \coordinate(m) at (0.5,-1);
        \coordinate(n) at (1.5,-1);
        \coordinate(o) at (2,-1);
        \coordinate(p) at (-0.75,-1);
        \coordinate(q) at (-0.25,-1);

        \node[style={rectangle, rounded corners, draw, fill=yellow!10}] (u) at (1,2.5) {\small{\textbf{$G$ Prime}}};
        \draw(v)node[below]{$a$} node{$\bullet$};
        \node[style={rectangle, rounded corners, draw, fill=blue!10}] (w) at (0.3,0.8) {\small{\textbf{S}}};
        \draw(x)node[below]{$e$} node{$\bullet$};
        \node[style={rectangle, rounded corners, draw, fill=blue!10}] (y) at (1.7,-0.2) {\small{\textbf{P}}};
        \draw(z)node[below]{$h$} node{$\bullet$};
        \node[style={rectangle, rounded corners, draw, fill=blue!10}] (l) at (-0.5,-0.2) {\small{\textbf{P}}};
        \draw(m)node[below]{$c$} node{$\bullet$};
        \draw(n)node[below]{$f$} node{$\bullet$};
        \draw(o)node[below]{$g$} node{$\bullet$};
        \draw(p)node[below]{$b$} node{$\bullet$};
        \draw(q)node[below]{$d$} node{$\bullet$};

        \draw (u) -- (v);
        \draw (u) -- (w);
        \draw (u) -- (x);
        \draw (u) -- (y);
        \draw (u) -- (z);
        \draw (w) -- (l);
        \draw (w) -- (m);
        \draw (y) -- (n);
        \draw (y) -- (o);
        \draw (l) -- (p);
        \draw (l) -- (q);
        
    \end{tikzpicture}
    \caption{A graph $G$ (left) and its modular decomposition tree (right). Maximal modules are red, series and parallel nodes are labelled in the tree as $S$ and $P$ respectively.}
    \label{fig:first_example}
\end{figure}
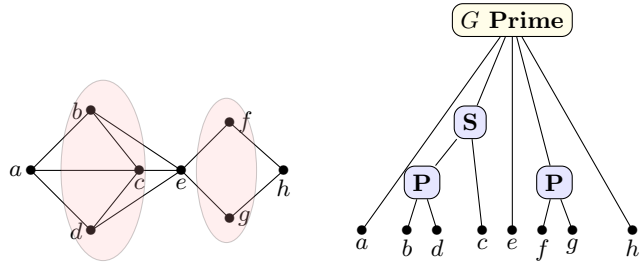

By the Modular Decomposition Theorem~\cite{CHM81,Gal67}, every graph admits a \emph{unique} modular decomposition tree. 
Other combinatorial objects also admit unique decomposition trees, partitive families in particular. 

Two sets $A$ and $B$ \emph{overlap} if $A \cup B \neq \emptyset$, $A \setminus B \neq \emptyset$, and $B \setminus A \neq \emptyset$. 
In a family of subsets $\cal F$ of a ground set $V$, a set $S \in \mathcal{F}$ is \emph{strong} if $S$ does not overlap with any other set in $\cal F$. 
We denote by $\Delta$  the \emph{symmetric difference} of two sets: 
$$A \Delta B = \{ a : a \in A \setminus B \} \cup \{ b : b \in B \setminus A \}.$$

\begin{definition}[\cite{CHM81}]
    \label{def:partitive_family}
    A family of subsets $\cal F$ over a ground set $V$ is \textbf{partitive} if
    \begin{itemize}
        \item [(i)] $\emptyset$, $V$, and all singletons $\{ x : x \in V \}$ belong to $\cal F$, and
        \item [(ii)] $\forall A, B \in \mathcal{F}$,
        if $A\cap B\neq\emptyset$ then 
        $A \cup B \in \mathcal{F}$, $A \cap B \in \mathcal{F}$, $A \setminus B \in \mathcal{F}$, and $A \Delta B \in \mathcal{F}$.
    \end{itemize}
\end{definition}

Partitive families play a fundamental role in combinatorial decomposition~\cite{Bui-XuanHR12,CHM81}. 
Every partitive family admits a unique decomposition tree with only complete and prime nodes. 
The strong elements of $\cal F$ form a tree ordered by the inclusion relation~\cite{CHM81}.

A \emph{complement reducible} graph is a graph whose decomposition tree has no prime nodes, that is, the graph is totally decomposable into parallel and series nodes only. 
Complement reducible graphs are also known as \emph{cographs}, and are exactly the $P_4$-free graphs~\cite{Seinsche74}. 
A modular decomposition tree of a cograph is often referred to as a \emph{cotree}. 
Cographs have been widely studied, and many typical $NP$-hard problems (colouring, independent set, etc.) become tractable on cographs~\cite{corneil1981complement}.

\subsection{Generalizations of Modular Decomposition / Motivation}

Finding a non-trivial tractable generalization of modules is not an easy task. 
Indeed, when trying to do so, we are faced with two main difficulties.

The first one is to obtain a pseudo-generalization.
Suppose for example that we change the definition of a module into:  $\forall x, y \in M$, $N^*(x) \setminus M=N^*(y) \setminus M$, where $N^*(x)$ can mean something like ``vertices at distance at most $k$'' or ``vertices joined by an odd path'', etc. 
In many of these scenarios, it turns out that the problem transforms itself into the computation of precisely the modules of some auxiliary graph built from the original one. 
Some work in this direction avoiding this drawback can be found in \cite{Bui-XuanHLM09}.

The second difficulty is $NP$-hardness. Consider the notion of \emph{roles} defined in sociology, where two vertices play the same role in a social network if they have the same set of colours in their neighbourhood. 
In this scenario, if a colouring of the vertices is given, then one can compute these \emph{roles} in polynomial time. Otherwise, the problem is indeed a colouring problem which is $NP$-hard to compute~\cite{FialaP05}.

In this work, we consider two variations of the notion of modules, both of which trying to avoid these two difficulties. Some of these new modules are polynomial to compute, and we believe they are worth studying further. We focus on the most promising relaxation, namely what we call $(\alpha, \beta)$-modules.

Our initial idea was to allow some ``errors'' by saying that at most $k$ edges (for some fixed integer $k$) could be missing in the complete bipartite subgraph between $M$ and $N(M)$, denoted $(M, N(M))$, and, symmetrically, that at most $k$ extra edges can exist between $M$ and $\overline{N}(M)$. 
But by doing so, we loose most of the nice algebraic properties of modules which yield an underlying partitive family. 
Furthermore, most modular decomposition algorithms are based on these algebraic properties~\cite{CHM81}. 

A second natural idea is to relax the condition on the complete bipartite subgraph $(M, N(M))$, for example by asking for a graph that does not contain any $2K_2$ (two disjoint edges). Unfortunately, as shown in~\cite{RS02}, to test whether a given graph admits such a decomposition is $NP$-complete. 
In fact, in the same work, the authors studied a generalized join decomposition solving a question raised in~\cite{Hsu87} about perfection. 

For all the above reasons and obstacles, we focus on $(\alpha, \beta)$-modules which maintain some algebraic properties and thus allow to obtain nice algorithms.

Intuitively, we want the reader to think of an $(\alpha, \beta)$-module as a subset of vertices that \emph{almost} looks the same from the outside. 
So, if $M$ is an $(\alpha, \beta)$-module, then for all $x, y \in M$, $N(x) \setminus M$ and $N(y) \setminus M$ are the same, with the exception of at most $\alpha + \beta$ ``errors'', where an error is either a missing edge or an extra edge. 
We use the integers $\alpha$ and $\beta$ to bound the number of errors in the adjacency, according to their type.

Formally, we define an $(\alpha, \beta)$-module as follows.

\begin{definition}
\label{def:alpha_beta_module}
    An \textbf{\boldmath$(\alpha, \beta)$-module} of a graph $G = (V, E)$ is a set of vertices $M \subseteq V$ that satisfies 
    \begin{align*}
            \forall x \in V \setminus M,\ |M \cap N(x)| \geq |M| - \alpha \text{ or }  |M \cap N(x)| \leq \beta.
    \end{align*}
\end{definition}

In other words, $M$ can be turned into a (standard) module by adding at most $\alpha$ edges or deleting at most $\beta$ edges at each vertex outside $M$.

This notion of missing or extra edges, that we call  \emph{$(\alpha,\beta)$-errors}, finds application naturally in various fields, from data compression and exact encodings to approximation algorithms.

Indeed, modular decomposition is often presented as an efficient way to encode a graph. 
This encoding property is preserved under the $(\alpha, \beta)$-modules. 

We want to be able to contract a non-trivial $(\alpha, \beta)$-module (to be precisely defined later, see Definition~\ref{def:trivial_alpha_beta_module}) into a single vertex while keeping \emph{almost} the entirety of the original graph, and then apply induction on the decomposition. 
To this end, for a graph $G = (V, E)$, let $M$ be a non-trivial $(\alpha, \beta)$-module with $X = N(M)$ and $Y = \overline{N}(M)$. 
If we want an exact encoding of $G$, we can contract $M$ into a unique vertex $m$ adjacent to every vertex in $X$, and non-adjacent to any vertex in $Y$. 
We then keep track of the subgraph $G(M)$ and the errors that potentially arose from the missing edges in $(M, X)$ and the extra edges in $(M, Y)$. 
This new encoding has at least $|X| \cdot (|M| - \alpha -1)$ edges  
less than the original encoding in the worst case and $|X| \cdot (|M| -1)$ when $M$ is a module.

A second natural and useful application of $(\alpha, \beta)$-modules concerns approximation algorithms. 
Similarly to how cographs are the totally decomposable graphs with respect to standard modular decomposition, we can define $(\alpha, \beta)$-cographs as the totally decomposable graphs with no $(\alpha, \beta)$-prime graphs (see Definition~\ref{def:alpha_beta_prime}). 
Now consider the classical colouring and independent set programs on cographs. 
The linear time algorithms for these problems both use modular decomposition. 
Roughly speaking, the algorithms compute a modular decomposition tree, and keep track of the series and parallel internal nodes of the cotree by scanning the tree from the leaves to the root. 
Now for $\alpha + \beta \leq 1$, we can get a simple 2-approximation algorithm for $(\alpha, \beta)$-cographs for both colouring and independent set, just by summing over all the $(\alpha,\beta)$- errors.

\section{Structural Properties of $(\alpha, \beta)$-Modules} 
\label{sec:structure}
In order to maintain some of the algebraic properties of modules, and avoid running into the $NP$-complete scenarios previously mentioned, the $(\alpha, \beta)$ generalization of modules seems to be a good compromise. 

We emphasize a few points concerning $(\alpha, \beta)$-modules. 
Note first that we tolerate $\alpha$ or $\beta$ ``error-edges'' per vertex outside the module, depending on how this vertex is connected to the $(\alpha, \beta)$-module, and not $\alpha + \beta)$ error-edges \emph{per module}.
Secondly, observe that when  $\alpha=\beta=0$, we recover the standard definition of modules (see Definition~\ref{def:moduleClassique}), which can be rephrased as follows.

\begin{definition}
\label{def:classic_2}
    A module of a graph $G = (V, E)$ is a set of vertices $M \subseteq V$ that satisfies 
    \begin{align*}
            \forall x \in V \setminus M,\ M \cap N(x) = \emptyset \text{ or } M \cap N(x) = M.
    \end{align*}    
\end{definition}

Of course we only consider cases for which $\max(\alpha, \beta) < |V| - 1$. 
Fig.~\ref{fig:second_example} illustrates an example of a graph with a $(1, 1)$-module. 

\begin{figure}[ht]
    \begin{center}
    \begin{tikzpicture}
        \coordinate(a) at (1,1);
        \coordinate(b) at (2,1);
        \coordinate(c) at (3,1);
        \coordinate(d) at (1,0);
        \coordinate(e) at (2,0);
        \coordinate(f) at (3,0);
        \coordinate(g) at (4,1);
        
        \draw(a)node[above]{$a$} node{$\bullet$};
        \draw(b)node[above]{$b$} node{$\bullet$};
        \draw(c)node[above]{$c$} node{$\bullet$};
        \draw(d)node[below]{$d$} node{$\bullet$};
        \draw(e)node[below]{$e$} node{$\bullet$};
        \draw(f)node[below]{$f$} node{$\bullet$};
        \draw(g)node[below]{$g$} node{$\bullet$};
        
        \draw (a)--(b)--(c);
        \draw (d)--(a);
        \draw (d)--(b);
        \draw (e)--(a);
        \draw (e)--(b);
        \draw (e)--(c);
        \draw (e)--(g);
        \draw (f)--(a);
        \draw (f)--(b);
        \draw (f)--(c);
    \end{tikzpicture}
    \end{center}
    \caption{The set $\{d, e, f\}$ is not a standard module, nor a $(1,0)$ or a $(0,1)$-module, only a $(1, 1)$-module.}
    \label{fig:second_example}
\end{figure}
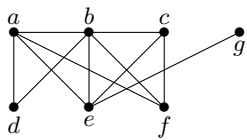

Let us begin with some simple properties that directly follow from Definition~\ref{def:alpha_beta_module}.

\begin{proposition} 
\label{prop:basic}
    If $M$ is an $(\alpha, \beta)$-module of $G$, then the following holds. 
    \begin{enumerate}
        \item\label{prop:basic_1} $M$ is an $(\alpha', \beta')$-module of $G$, for every $\alpha \leq \alpha'$ and $\beta \leq \beta'$. 
        \item\label{prop:basic_2} $M$ is a $(\beta, \alpha)$-module of $\overline{G}$. 
        \item\label{prop:basic_3} $M$ is an $(\alpha, \beta)$-module of every induced subgraph $G(N)$ of $G$ with $M \subseteq N$.
        \item\label{prop:basic_4}\label{recursive} Every $(\alpha, \beta)$-module of $G(M)$ is an $(\alpha, \beta)$-module of $G$. 
    \end{enumerate}
\end{proposition}

\begin{proof}\mbox{}
\begin{enumerate}
\item Taking $\alpha'$ and $\beta'$ such that $\alpha \leq \alpha'$ and $\beta \leq \beta'$ can only relax the module conditions.
\item Moving to the complement just interchanges the roles of $\alpha$ and $\beta$ in the definition.
\item If the $(\alpha, \beta)$-module conditions are satisfied for all vertices in $V(G) \setminus M$, then they are satisfied for all vertices in $V(G) \setminus N$ for $M \subseteq N$. Therefore $M$ is an $(\alpha, \beta)$-module of the induced subgraph $G(N)$.
\item
Let $N$ be an $(\alpha, \beta)$-module of the subgraph $G(M)$. So every vertex in $M \setminus N$ satisfies the $(\alpha, \beta)$-module conditions. Since $M$ is supposed to be an $(\alpha, \beta)$-module,  every vertex in $V(G) \setminus M$ satisfies the $(\alpha, \beta)$-module conditions for $M$ and therefore also for $N \subseteq M$.
\end{enumerate}
\end{proof}

\begin{definition}
\label{def:alpha_beta_voisinage}

    Let $G = (V, E)$ be a graph and $A \subseteq V$ be a set of vertices.
    The \textbf{\boldmath$\alpha$-neighbourhood} and \ERIC{\textbf{\boldmath$\beta$-non-neighbourhood}} of $A$ are, respectively,
    \begin{align*}
        N_{\alpha}(A) &= \{ x \notin A : |N(x) \cap A| \geq |A|-\alpha \}, \text{ and}
        \\
        \overline{N}_{\beta}(A) &= \{ x \notin A : |N(x) \cap A| \leq \beta \}.
    \end{align*}
    \ERIC{Moreover, if $x\in N_\alpha(A)$ (resp. $x\in \overline{N}_{\beta}(A)$), we say that $x$ is an \textbf{\boldmath$\alpha$-neighbour} of $A$ (resp. a \textbf{\boldmath$\beta$-non-neighbour} of $A$)
    and that $x$ is \textbf{\boldmath$\alpha$-adjacent} (resp. \textbf{\boldmath$\beta$-non-adjacent} to every vertex of $A$.
   }
    \end{definition}

\begin{definition}
\label{def:alpha_beta_splitter}
    Let $G = (V, E)$ be a graph and $A \subseteq V$ be a set of vertices. 
    A vertex $z \notin A$ is an \textbf{\boldmath$(\alpha, \beta)$-splitter} for $A$ if
    \begin{align*}
        \beta <|N(z)\cap A|< |A|-\alpha.
    \end{align*}
    We denote by $S_{\alpha, \beta}(A)$ the set of $(\alpha, \beta)$-splitters of $A$.
\end{definition}

Hence, a set $A$ is an $(\alpha, \beta)$-module if and only if $S_{\alpha, \beta}(A)= \emptyset$. 
As an immediate consequence we have the following easy facts.

\begin{lemma}
\label{lem:comptages}
    For every graph $G=(V,E)$ and every set of vertices $A \subseteq V$, the following holds.
        \begin{enumerate} 
        \item\label{lem:comptages.1} $N_{\alpha}(A) \cup \overline{N}{_{\beta}(A)} \cup S_{\alpha, \beta}(A) = V\setminus A$. 
        \vspace{0.2cm}
        \item\label{lem:comptages.2} If $|A| \geq \alpha + \beta +1$, 
            then  $N_{\alpha}(A) \cap \overline{N}_{\beta}(A)=\emptyset$.
        \vspace{0.2cm}
        \item\label{lem:comptages.3} If $|A| \leq \alpha + \beta +1$, 
            then $S_{\alpha, \beta}(A)=\emptyset$.
        \vspace{0.2cm}
        \item\label{lem:comptages.4} If $|A| = \alpha + \beta +1$,
           then $N_{\alpha}(A)$ and $\overline{N}_{\beta}(A)$ partition $V\setminus A$.
        \vspace{0.2cm}        
        \item\label{lem:comptages.5} If $A$ is an $(\alpha, \beta)$-module of $G$ and $|A| \geq \alpha + \beta +1$, then $N_{\alpha}(A)$ and  $\overline{N}_{\beta}(A)$ partition $V\setminus A$.
    \end{enumerate}

\end{lemma}

\begin{proof}\mbox{}
    \begin{enumerate}
    \item This directly follows from the definitions of these sets. 
    \item If $x \in N_{\alpha}(A)$, then $|N(x) \cap A| \geq |A|-\alpha \geq \beta +1$ and thus $x \notin \overline{N}_{\beta}(A)$.
    \item If $x \in S_{\alpha, \beta}(A)$, then $|N(z)\cap A|< |A|-\alpha <\beta+1$, a contradiction.
    \item We have $N_{\alpha}(A) \cap \overline{N}_{\beta}(A)=\emptyset$ by Item~2, and $S_{\alpha, \beta}(A) =\emptyset$ by Item~3. The result then follows from Item~1. 

    \item This follows from Items 1 and~2.
    \qed
    \end{enumerate}
\end{proof}

\begin{lemma}
\label{lem:splitter}
For every graph $G=(V,E)$ and every set of vertices $A \subseteq V$, if
    $|A| \leq \alpha+\beta+1$, then $A$ is an $(\alpha, \beta)$-module of $G$.
\end{lemma}

\begin{proof}
    Using Lemma~\ref{lem:comptages}.\ref{lem:comptages.3}, $A$ admits no $(\alpha, \beta)$-splitter and  is thus an $(\alpha, \beta)$-module of $G$. 
\qed
\end{proof}

It thus seems that the subsets of size $\alpha + \beta +1$ are crucial to the study of this new decomposition. 
In fact, if $A$ is such a set, then for every vertex $z \notin A$, we have either $z \in N_{\alpha}(A)$ or $z \in \overline{N}_{\beta}(A)$, but not both (Lemma~\ref{lem:comptages}.\ref{lem:comptages.3}). 

\begin{lemma}
\label{lem:splitf}
    If a vertex $s$ is an $(\alpha,\beta)$-splitter for a set $A$, then $s$ is also an $(\alpha,\beta)$-splitter for every set $B \supseteq A$ with $s \notin B$.
\end{lemma}

\begin{proof}
Let $s$ be an $(\alpha, \beta)$-splitter of $A$.
We thus have $\beta < |N(s) \cap A| < |A| - \alpha$. 
Now, if $A \subseteq B$ and $s \notin B$, then 
we have $ \beta < |N(s)\cap A| \leq |N(s)\cap B|$.
Similarly, since $N(s)\setminus B \supseteq N(s)\setminus A$, we have $|N(s)\cap B| < |B| - \alpha$.
Therefore, we get $\beta < |N(s)\cap B| < |B| - \alpha$ and $s$ is an $(\alpha, \beta)$-splitter for $B$. 
\qed 
\end{proof}

\begin{theorem}
\label{thm:prepar}
For every graph $G = (V, E)$, the family of $(\alpha, \beta)$-modules of $G$ satisfies the following.
   \begin{enumerate}
     \item
        The set $V$ is an $(\alpha, \beta)$-module of $G$, and 
        every set $A \subseteq V$ with $|A| \leq  \alpha+\beta+1$ is an $(\alpha, \beta)$-module of $G$.         
     \item
        If $A$ and $B$ are two $(\alpha, \beta)$-modules of $G$, then $A \cap B$ is an $(\alpha, \beta)$-module of $G$. Moreover, the $(\alpha, \beta)$-splitters of $A \setminus B$ and $B \setminus A$ can only belong to $A \cap B$. 
    \end{enumerate}

\end{theorem}

\begin{proof}\mbox{}
    \begin{enumerate}
        \item This directly follows from Definition~\ref{def:alpha_beta_module} and Lemma~\ref{lem:splitter}. 
        \item
        First notice that if both $A$ and $B$ 
        are $(\alpha,\beta)$-modules considered in Item~1,
        then  so do $A \cap B$, $A \setminus B$ and $B \setminus A$. 
        Suppose then that the cardinality of both $A$ and $B$ is at least $\alpha+\beta+2$ and at most $|V|-1$.

        If $A \cap B$ has an $(\alpha,\beta)$-splitter outside of $A \cup B$ then, by Lemma~\ref{lem:splitf}, $A$ and $B$ would also have an $(\alpha,\beta)$-splitter, a contradiction. 
        If $A \cap B$ has an $(\alpha,\beta)$-splitter in $B\setminus A$ (resp. in $A\setminus B$) then, by Lemma~\ref{lem:splitf}, again $A$ (resp. $B$) would have an $(\alpha, \beta)$-splitter. Therefore, $A\cap B$ is an $(\alpha, \beta)$-module of $G$.

        Let us now consider $A \setminus B$.
        If $A \setminus B$ has an $(\alpha,\beta)$-splitter in $B \setminus A$ then, by Lemma~\ref{lem:splitf}, $A$ would have a $(\alpha,\beta)$-splitter as well. 
        The same conclusion arises for splitters outside of $A \cup B$.
        Hence, the only possible $(\alpha,\beta)$-splitters for $A \setminus B$ and, similarly, for $B \setminus A$, are in $A \cap B$.
        \qed 
    \end{enumerate}   
\end{proof}

Since the family of $(\alpha, \beta)$-modules is closed under intersection, it yields a notion of \textbf{graph convexity}. 
Given a set $A$, we can compute the minimal (under inclusion) $(\alpha, \beta)$-module $M(A)$ that contains $A$, with strictly more than $\alpha + \beta + 1$ elements, thus computing a \textbf{modular closure} via $(\alpha, \beta)$-splitters. 
Furthermore, the dual cases of $(1,0)$-modules and $(1,0)$-modules seem very interesting.

\begin{definition}
\label{def:trivial_alpha_beta_module}
    An $(\alpha,\beta)$-module $M$ of a graph $G=(V,E)$ is  a \textbf{trivial \boldmath$(\alpha,\beta)$-module} if either $M=V$ or $|M|\leq \alpha+\beta+1$.
\end{definition}

\begin{definition}
\label{def:alpha_beta_prime}
    A graph is an \textbf{\boldmath$(\alpha,\beta)$-prime graph} if it has only trivial $(\alpha,\beta)$-modules.
\end{definition}

Observe here that when $\alpha=\beta=0$, trivial $(\alpha,\beta)$-modules are exactly trivial (standard) modules, and $(\alpha,\beta)$-prime graphs are exactly prime graphs.

From Lemma~\ref{lem:splitter}, we directly get the following result.

\begin{corollary}
\label{cor:degenerate}
    A graph $G = (V, E)$ with $|V| \leq \alpha +\beta +1 $ has only trivial $(\alpha,\beta)$-modules.
\end{corollary}

However, we want to distinguish ``truly'' $(\alpha, \beta)$-prime graphs and ``degenerate''
$(\alpha, \beta)$-prime graphs.

\begin{definition}
\label{def:alpha_beta_minimal}
    A graph $G = (V, E)$ is \textbf{\boldmath$(\alpha, \beta)$-degenerate} if $|V| \leq \alpha + \beta + 2$. 
\end{definition}

Let us say that a non-trivial $(\alpha, \beta)$-module $A$ is a \textbf{minimal non-trivial \boldmath$(\alpha, \beta)$-module} if every $(\alpha, \beta)$-module strictly contained in $A$ is trivial.
The following result directly follows from this definition.

\begin{proposition}
\label{prop:minimal}
     If $A$ and $B$ are overlapping minimal non-trivial $(\alpha,\beta)$-modules of a graph $G$,
     then $A \cap B$ is a trivial $(\alpha,\beta)$-module of $G$.
\end{proposition}

From Theorem~\ref{thm:prepar}, we get the following result.

\begin{corollary}
\label{cor:moore_family}
    For every graph $G$, the family of $(\alpha, \beta)$-modules of $G$ is a Moore family. 
\end{corollary}

In the standard setting, if $X$ and $Y$ are modules, then $X \cup Y$, $X \setminus Y$, and $Y \setminus X$ are also modules~\cite{CHM81,HabibP10}. 
Unfortunately, this does not always hold in the $(\alpha, \beta)$ setting. 
But we can prove a weaker result, namely that we still have an ``almost partitive'' family. 

\begin{theorem}
\label{thm:presquepar}
    Let $A$ and $B$ be two non-trivial overlapping $(\alpha,\beta)$-modules of a graph $G$.    
    If $|A \cap B| \geq \alpha +\beta +1 $, then $A \cup B$ and $A \Delta B$ are $(2\alpha, 2\beta)$-modules of~$G$.
\end{theorem}

\begin{proof}
    Let $z \in V\setminus B$. 
We have $S_ {\alpha, \beta}(B)=\emptyset$
since $B$ is an $(\alpha,\beta)$-module, and $|B| \geq \alpha + \beta+1$ since $B$ is non-trivial.
    Therefore, by Lemma~\ref{lem:comptages}.\ref{lem:comptages.5}, $N_{\alpha}(B)$ and $\overline{N}_{\beta}(B)$ partition $V\setminus B$. 
    Suppose $z \in N_{\alpha}(B)$. Then, $z$ has at most $\alpha$ non-neighbours in $A\cap B$
    and thus at least $\beta + 1$ neighbours in $A \cap B$. 
    This gives $z \in N_{\alpha}(A)$ since $A$ is an $(\alpha,\beta)$-module.

Consider first $A \cup B$. 
In the worst case, $z$ has at most $\alpha$ non-neighbours in $A\setminus B$ and at most $\alpha$ non-neighbours in $B\setminus A$.
Using the same reasoning on $\overline{N}_{\beta}(B)$, we get that $A \cup B$ is a $(2\alpha, 2\beta)$-module.
        
Consider now $A \Delta B$. We just have to further consider the case of vertices in $A\cap B$. Also in the worst case, errors arise when a given vertex $z \in A\cap B$ has $\alpha$ (resp. $\beta$) errors in $A \setminus B$ and $\alpha$ (resp. $\beta$) errors in $B\setminus A$. Therefore $A \Delta B$ is a $(2\alpha, 2\beta)$-module. \qed

\end{proof}

\subsection{$(\alpha,\beta)$-Modular Decomposition Trees}

\ERIC{

\begin{definition}
Let $G=(V,E)$ be a graph. Two disjoint sets of vertices $A,B\subseteq V$,
with $|A|,|B| \geq \alpha + \beta + 1$,
 are said to be \textbf{\boldmath$\alpha$-connected} if $A\subseteq N_\alpha(B)$ and $B\subseteq N_\alpha(A)$.
 Similarly, they are said to be \textbf{\boldmath$\beta$-non-connected} if $A\subseteq \overline{N}_\beta(B)$ and $B\subseteq \overline{N}_\beta(A)$.
\label{def:alpha-beta-connected}
\end{definition}

In other words, $A$ and $B$ are $\alpha$-connected if every vertex in $A$ is an $\alpha$-neighbour of $B$ and every vertex in $B$ is an $\alpha$-neighbour of $A$.
They are $\beta$-non-connected if every vertex in $A$ is a $\beta$-non-neighbour of $B$ and every vertex in $B$ is a $\beta$-non-neighbour of $A$.

}

\begin{definition}
\label{def:SP}
Let $G = (V, E)$ be a graph. 
\begin{itemize}
    \item An \textbf{\boldmath$(\alpha, \beta)$-modular partition} of $G$ is a partition of $V$ into $(\alpha, \beta)$-modules.
    
        \vspace{0.2cm}    
        \item For $|V| \geq \alpha + \beta + 3$, we say that $G$ admits an \textbf{\boldmath$\alpha$-series} (resp. a \textbf{\boldmath$\beta$-parallel}) \textbf{decomposition} if there exists an $(\alpha, \beta)$-modular partition of $V$, $\mathcal{P} = \{V_1, \ldots, V_k\}$, such that 
        \ERIC{
        \begin{enumerate}
           \vspace{0.2cm}
           \item 
           $\exists j \in [k]$ such that $|V_j| \geq \alpha + \beta + 1$, and
           \vspace{0.2cm}
            \item $\forall i,j \in [k]$,  $i \neq j$, $V_i$ and $V_j$ are $\alpha$-connected (resp. $\beta$-non-connected).

        \end{enumerate} 
        }
    
        \vspace{0.2cm}
        \item For $|V| \geq \alpha + \beta + 3$, we say that $G$ admits an \textbf{\boldmath$(\alpha, \beta)$-prime decomposition} if there exists an $(\alpha, \beta)$-modular partition of $V$, $\mathcal{P} = \{V_1, \ldots, V_k\}$, with $k \geq 2$ such that 
        \ERIC{
    \begin{enumerate}  
    \vspace{0.2cm}
    \item  $\forall i \in [k]$, $V_i$ is maximal (under inclusion), 
    
    \vspace{0.2cm}
    \item $\exists j \in [k]$ such that $|V_j| \geq \alpha + \beta + 1$, and
    \vspace{0.2cm}
    \item  there exist two pairs $(i,j)$ and $(p,q)$, $i\neq j$, $p\neq q$, such that $V_i$ and $V_j$ are $\alpha$-connected while $V_p$ and $V_q$ are $\beta$-non-connected. 
    \end{enumerate}
    }
\end{itemize}
\end{definition}

Using Proposition~\ref{prop:basic}.\ref{prop:basic_2} we have the following obvious property.

\begin{property}\label{SP}
\label{Series-Parallel}
    A graph $G$ admits an $\alpha$-series decomposition if and only if $\overline{G}$ admits an $\alpha$-parallel decomposition. 
\end{property}

\begin{proposition}
\label{prop:exclusive}
    If $A$ and $B$ are two disjoint $(\alpha, \beta)$-modules of a graph $G$ with $|A|, |B| \geq \alpha+\beta +1$, then
        $N_{\alpha}(A) \supseteq B$ and  $\overline{N}_{\beta}(B) \supseteq A$ are mutually exclusive.
\end{proposition}

\begin{proof}
    Suppose to the contrary that $N_{\alpha}(A) \supseteq B$ and 
   $\overline{N}_{\beta}(B) \supseteq A$.
   
    Let $m_{A,B}$ be the number of edges in $G$ joining $A$ and $B$.
    \ERIC{
    We thus have
        $$|B| \cdot (|A|- \alpha) \leq m_{A,B} \leq |A| \cdot \beta \implies |A|\cdot |B| < \beta \cdot |A|  +\alpha \cdot |B|.$$
    Now, let $|A|=\alpha +\beta +a$ and $|B|=\alpha +\beta +b$, with $a, b \in \mathbb{N}^*$. We then get
    $$\alpha \cdot a  + \beta \cdot b + ab \leq 0,$$
    a contradiction.
    }
\qed 
\end{proof}

Furthermore it could be the case that $A$ is $\alpha$-connected to $B$ and that $B$ is neither $\alpha$-connected to $A$ nor $\beta$-non-connected to $A$.

\begin{corollary}\label{ouf}
If $A$ and $B$ are two disjoint $(\alpha, \beta)$-modules of a graph $G$ with $|A| \geq \alpha+\beta +1$ and $\alpha \geq 1$, then the inclusions
        $N_{\alpha}(A) \supseteq B$ and  $\overline{N}_{\beta}(B) \supseteq A$ are mutually exclusive.
        
         If $A$ and $B$ are two disjoint $(\alpha, \beta)$-modules of a graph $G$ with $|B| \geq \alpha+\beta +1$ and $\beta \geq 1$, then the inclusions
        $N_{\alpha}(A) \supseteq B$ and  $\overline{N}_{\beta}(B) \supseteq A$ are mutually exclusive.

\end{corollary}

\begin{proof}
\ERIC{If $|A| \geq \alpha+\beta +1$ and $|B| \geq \alpha+\beta +1$, this directly follows from Proposition~\ref{prop:exclusive}.
Otherwise, the proof is similar to the proof of Proposition~\ref{prop:exclusive}, by simply
considering the two extreme cases, i.e., $b \leq 0$, $\alpha \geq 1$ and  $a \geq 1$, or
$a \leq 0$, $\beta \geq 1$ and $b \geq 1$.}
\end{proof}

\begin{proposition}
    The three cases of Definition~\ref{def:SP} are mutually exclusive for a given partition into $(\alpha, \beta)$-modules.
\end{proposition}

\begin{proof}
    Clearly the $\alpha$-series and $\beta$-parallel decompositions are each exclusive with the $(\alpha, \beta)$-prime case. 
    
    Now let us prove  that $\alpha$-series and $\beta$-parallel cases are exclusive. If $\alpha=\beta=0$ it is just Gallai's theorem. 
    So we suppose that there exists an $(\alpha, \beta)$-modular partition ${\cal P}=\{V_1, \dots, V_k\}$ which is both an $\alpha$-series and a $\beta$-parallel decomposition.
Suppose, without loss of generality, that $|V_1| \geq \alpha +\beta +1$ (we know that such a set exists).
Now we have two cases to consider.

If $\alpha \geq 1$ then, using Corollary \ref{ouf},  if $V_1$ is $\alpha$-connected to $V_2$   
then $V_2$ cannot be $\beta$-non-connected to $V_1$.

Similarly, if $\beta \geq 1$,   again using Corollary \ref{ouf},  if $V_1$ is $\beta$-non-connected to $V_2$ then $V_2$ cannot be $\alpha$-connected to $V_1$.
\qed    
\end{proof}

\begin{mdframed}[style=MyFrame]
\begin{openproblem} 
Is this result also true for $\alpha$-series and $\beta$-parallel decompositions based on different partitions?
(The interesting case is when all intersections between the two partitions have size bounded by $\alpha +\beta$.)
\end{openproblem} 
\end{mdframed}

It should be noticed that in the case of an $\alpha$-series  (resp.  a $\beta$-parallel) decomposition, a union of $V_i$'s is not necessarily an $(\alpha, 0)$-module  (resp. a $(0, \beta)$-module). 
Such a property is always true only for standard modules.

\begin{definition}
\label{def:alpha_beta_brittle}
    Using the terminology of~\cite{CE80} for combinatorial decompositions, we will say that
    a graph $G=(V,E)$ is \textbf{\boldmath$(\alpha, \beta)$-brittle}
    if every subset of $V$ is an $(\alpha, \beta)$-module.
\end{definition}

Of course, $(\alpha, 0)$-complete graphs (i.e., complete graphs missing at most $\alpha$ edges), and $(0, \beta)$-independent graphs (i.e., independent sets with at most $\beta$ edges) are $(\alpha, \beta)$-brittle, but they are not the only obvious ones; any path $P_k$ is also $(1,1)$-brittle.

As already seen previously, all graphs $G$ with $|V|\leq \alpha +\beta +2$ are $(\alpha, \beta)$-brittle,
and we called them $(\alpha, \beta)$-degenerate to distinguish them from the ``truly'' $(\alpha, \beta)$-prime graphs. 
All these remarks raise the question of the characterization of $(\alpha, \beta)$-brittle graphs. 
Clearly, any graph $G$ \ERIC{with minimum degree at least $|V| - \alpha$ or maximum degree at most $\beta$}
is $(\alpha, \beta)$-brittle.

\begin{mdframed}[style=MyFrame]
\begin{openproblem} 
Can we characterize the $(\alpha, \beta)$-brittle graphs?
\end{openproblem} 
\end{mdframed}

Using Definition~\ref{def:SP} and Proposition~\ref{prop:basic}.\ref{prop:basic_4}, and mimicking the case of standard modular decomposition, we may define an \ERIC{$(\alpha, \beta)$-modular decomposition tree} as follows.

\begin{definition}
\label{def:alpha_beta_decomp_tree}
    An \textbf{\boldmath$(\alpha, \beta)$-modular decomposition tree} of a graph $G=(V,E)$ is a tree whose nodes are labelled with $(\alpha, \beta)$-modules ordered by inclusion with four types of nodes:
    \begin{enumerate}
        \item $\alpha$-series,
        \item $\beta$-parallel,
        \item $(\alpha, \beta)$-prime, and
        \item$(\alpha, \beta)$-degenerate.
    \end{enumerate} 
    
    Each level of the tree corresponds to a partition of $V$, starting with $\{V\}$ at the root, and such that the leaves correspond to a partition of $V$ into $(\alpha, \beta)$-degenerate nodes.
\end{definition}

\ERIC{
Such an $(\alpha, \beta)$-modular decomposition tree completely describes what we call an \textbf{\boldmath$(\alpha,\beta)$-modular decomposition} of a graph $G$.
}

In the standard modular decomposition setting, the notion of strong modules, i.e., modules that do not overlap any other module, is quite central. 
In the $(\alpha,\beta)$-modular decomposition setting, observe that there are no strong $(\alpha, \beta)$-modules other than $\{V\}$ and the singletons $\{ v : v \in V\}$. 
This comes from the fact that when $\max\{\alpha, \beta\} \geq 1$, every subset of vertices of size $2$ is a trivial $(\alpha, \beta)$-module. 
Now, assume there is a standard strong module $A \neq V$ with $|A| > 1$.
By taking  any vertex $v\in A$ and any vertex $u\in V \setminus A$, we get an $(\alpha, \beta)$-module of size~2 which overlaps~$A$. 

Recall that the Modular Decomposition Theorem~\cite{Gal67} (Gallai's Theorem) says that every graph admits a unique modular decomposition tree. 
In the $(\alpha, \beta)$ setting, we have the following weaker form of Gallai's Theorem.

\begin{theorem}
    \label{thm:modulartree}
    Every graph $G = (V, E)$ with $|V| \geq \alpha + \beta +3$ admits an $(\alpha, \beta)$-modular decomposition tree, for every $(\alpha, \beta)$ with $0 \leq \alpha, \beta \leq |V| - 1$.
\end{theorem}

\begin{proof}
Let $G = (V, E)$ be an arbitrary graph. 
If $G$ is an $(\alpha,\beta)$-prime graph then $G$ admits a trivial $(\alpha,\beta)$-modular decomposition tree $T(G)$, with a root labelled $(\alpha, \beta)$-prime. The leaves of $T(G)$ are the children of the root, where every leaf is associated with a partition of $V$ into sets of $\alpha + \beta +1$ vertices.

Suppose now that $G$ is not an $(\alpha,\beta)$-prime graph.
Then $G$ admits at least one non-trivial $(\alpha,\beta)$-module. 
Let $M_1$ be any maximal (under inclusion) non-trivial $(\alpha,\beta)$-module, and let $R$ be the set of remaining vertices: $R = V \setminus M_1$.
We consider two cases, depending on the size of $R$.

\begin{enumerate}
\item $|R| \leq \alpha+\beta$. \\
If $R=N_{\alpha}(M_1)$, then we build a tree with a root $r$ labelled $(\alpha,\beta)$-series, and two children, $T(G(M1))$ and $R$ (labelled $(\alpha,\beta)$-degenerate).
    
If $R= \overline{N}_{\beta}(M_1)$, then we build a tree with a root $r$ labelled $(\alpha,\beta)$-parallel and two children, $T(G(M1))$ and $R$ (labelled $(\alpha,\beta)$-degenerate).
    
Finally, in the remaining case, both $N_{\alpha}(M_1)$ and $\overline{N}_{\beta}(M_1)$ are non-empty. 
We then build a tree with a root $r$ labelled $(\alpha, \beta)$-prime and three children, $T(G(M1))$,  $N_{\alpha}(M_1)$ and $\overline{N}_{\beta}(M_1)$, 
the last two nodes being both labelled $(\alpha,\beta)$-degenerate.

We then apply induction on $G(M_1)$ to obtain $T(G(M1))$, using Proposition~\ref{prop:basic}.4.

\item $|R| \geq \alpha + \beta +1$.\\
We simply apply the same reasoning on $G(V \setminus M_1)$. 
\end{enumerate}

In both cases, we eventually obtain a partition ${\cal P}=\{M_1, \dots, M_k\}$ of $V$ with $k \geq 3$. 
The partition $\cal P$ is an $(\alpha,\beta)$-modular partition in which only $M_k$ could have less than $\alpha + \beta$ vertices.

We then check whether $\cal P$ induces an $\alpha$-series decomposition  or a $\beta$-parallel decomposition. If not, we say that $\cal P $ induces an $(\alpha,\beta)$-prime decomposition.

We then 
build a tree with a root labelled $(\alpha, \beta)$-prime (resp.  $\alpha$-series,  $\beta$-parallel) and whose children are the $(\alpha,\beta)$-modular decomposition trees of the subgraphs $G(M_i)$, $1 \leq i \leq k$, and apply induction on the $G(M_i)$'s, using Proposition~\ref{prop:basic}.4. 
\qed
\end{proof}

Maximal $(\alpha, \beta)$-modules may overlap, which unfortunately means, as it is also illustrated in the examples of the next sections, that this $(\alpha, \beta)$-modular decomposition tree is not unique. 
Although such a decomposition tree can provide an exact encoding of the graph (if the $(\alpha,\beta)$-errors are traced), it does not provide an encoding of all existing $(\alpha,\beta)$-modules, see Figure~\ref{1MD-exp} for instance.
Furthermore, although the above proof is constructive, Theorem~\ref{thm:modulartree} does not lead to an efficient algorithm for computing an $(\alpha, \beta)$-modular decomposition tree. 
In fact, we do not know any polynomial algorithm to compute $M_1$, a maximal non-trivial $(\alpha, \beta)$-module, to start with.
In the next section we will see that it could be the case that no such algorithm exists.

\subsection{$\alpha$-Series and $\beta$-Parallel Operations} 
\label{sec:cograph}
Let us focus on the $\alpha$-series and $\beta$-parallel decompositions.

\begin{definition}
\label{def:alpha_beta_cograph}
    An \textbf{\boldmath$(\alpha, \beta)$-cograph} is a graph that is totally decomposable with respect to $\alpha$-series and $\beta$-parallel decompositions until we reach $(\alpha, \beta)$-degenerate subgraphs only. 
\end{definition}

Using Definition~\ref{def:alpha_beta_cograph} above, it is clear that standard cographs are precisely the $(0,0)$-cographs.
Let us call \textbf{\boldmath$(\alpha, \beta)$-cotree} the tree corresponding to an $(\alpha, \beta)$-modular decomposition  of an $(\alpha, \beta)$-cograph. 
\ERIC{Although several different $(\alpha, \beta)$-cotrees may be associated with an $(\alpha, \beta)$-cograph, the following proposition shows that we can always find such an $(\alpha, \beta)$-cotree having only $\alpha$-series and $\beta$-parallel internal nodes.

\begin{proposition}
    A graph is an $(\alpha, \beta)$-cograph if and only if it admits an $(\alpha, \beta)$-cotree having only $\alpha$-series and $\beta$-parallel internal nodes.
\end{proposition}
}

\begin{proof}
    We construct the $(\alpha, \beta)$-cotree of an $(\alpha, \beta)$-cograph $G$ recursively as follows.
    Suppose $G$ has an $\alpha$-series and $\beta$-parallel decomposition with partition ${\cal P}=\{V_1, \dots V_k\}$.
    Notice first that these two operations are exclusive since every part $V_i$ has at least $\alpha + \beta + 1$ vertices, and thus two parts $V_i$ and $V_j$ cannot be both $\alpha$-connected and $\beta$-non-connected. 
    The  $(\alpha, \beta)$-cotree is then obtained by taking a root $r$, labelled $\alpha$-series, with $k$ children $T(G(V_i))$, $1\le i\le k$, and construct each subtree $T(G(V_i))$ by applying induction on the     subgraph $G(V_i)$, using Proposition~\ref{prop:basic}.4.
    \qed
\end{proof}

Consider the two examples illustrated in Figures~\ref{1MD-exp} and~\ref{two-diff-decomps}. 
Fig.~\ref{1MD-exp} shows a $(1,1)$-cograph $H$ that admits a unique $(1,1)$-cotree. 
Fig.~\ref{two-diff-decomps} shows a $(1,1)$-cograph $G$ that admits \emph{two} different $(1,1)$-cotrees. 
Moreover, if we replace each vertex of $G$ in Fig.~\ref{two-diff-decomps} by an isomorphic copy of $G$, and repeat this process, we can build a $(1,1)$-cograph which has \emph{exponentially many} different  $(1,1)$-cotrees.

\begin{figure}
\centering
    \begin{minipage}{.9\textwidth}
    \centering
        \begin{tikzpicture}[scale=0.7,auto=left]
        	\node[style={circle,inner sep=1pt, minimum width=4pt,draw,fill=black!100},label=above:$a$] (a) at  (1.5,0) {};
        	\node[style={circle,inner sep=1pt, minimum width=4pt,draw,fill=black!100},label=above:$b$] (b) at  (0.5, 1) {};
        	\node[style={circle,inner sep=1pt, minimum width=4pt,draw,fill=black!100},label=below:$c$] (c) at  (0.5, -2) {};
        	\node[style={circle,inner sep=1pt, minimum width=4pt,draw,fill=black!100},label=below:$d$] (d) at  (1.5,-1) {};
        	
        	\node[style={circle,inner sep=1pt, minimum width=4pt,draw,fill=black!100},label=above:$e$] (e) at  (-3,0) {};
        	\node[style={circle,inner sep=1pt, minimum width=4pt,draw,fill=black!100},label=above:$f$] (f) at  (-2,1) {};
        	\node[style={circle,inner sep=1pt, minimum width=4pt,draw,fill=black!100},label=below:$h$] (h) at  (-3,-1) {};
        	\node[style={circle,inner sep=1pt, minimum width=4pt,draw,fill=black!100},label=below:$g$] (g) at  (-2,-2) {};
        
        	\foreach \from/\to in {a/b, b/c, a/e,c/d, e/f, f/g, g/h, b/e,b/f, a/g, b/h, c/e, a/f, c/h, c/g, d/e, d/f, d/h, d/g}
        	\draw (\from) -- (\to);
        \end{tikzpicture}
\hskip 5mm
        \begin{tikzpicture}[scale=0.7,auto=left]
        	\node[style={rectangle, rounded corners, draw},label=below:{},align=center] (v) at (-0.4,3) {{\color{black}(1,0)\text{-series}}};
        	\node[style={rectangle, rounded corners, draw},label=below:{},align=center] (u) at (2, 1.5) {{\color{black}(0,1)\text{-par.}}};
	\node[style={rectangle, rounded corners, draw},label=below:{},align=center] (t) at (-2.5,1.5) {{\color{black}(0,1)\text{-par.}}};

	\node[style={rectangle, rounded corners, draw},label=below:{\{a, b\}}] (w) at (-4,0) {\small{(0,1)-deg.}};
        	\node[style={rectangle, rounded corners, draw},label=below: {\{c, d\}}] (x) at (-1.5,0) {\small{(0,1)-deg.}};
	\node[style={rectangle, rounded corners, draw},label=below:{\{e, f \}}] (y) at (1,0) {\small{(0,1)-deg.}};
        	\node[style={rectangle, rounded corners, draw},label=below: {\{g, h\}}] (z) at (3.5,0) {\small{(0,1)-deg.}};
        	\foreach \from/\to in {v/u,v/t}
        	\draw (\from) -- (\to);
	\foreach \from/\to in {t/w, t/x}
	\draw (\from) -- (\to);
	\foreach \from/\to in {u/y,u/z}
	\draw (\from) -- (\to);
	\end{tikzpicture}
    \end{minipage}
    \caption{    
        A $(1,1)$-cotree of the $(1,1)$-cograph $H$ on the left. 
        Notice that $H$ is not a cograph since it contains two induced $P_4$'s:
        $\{a, b, c, d\}$ and $\{e, f, g, h\}$.
    }
    \label{1MD-exp}

\end{figure}
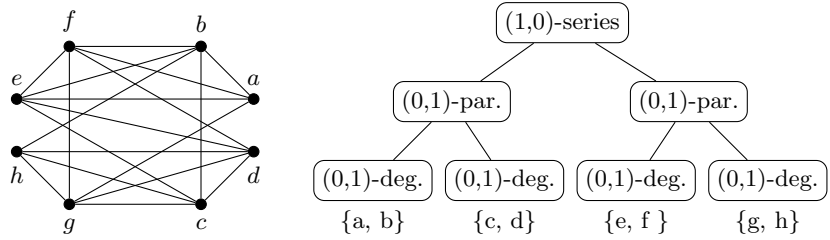

\begin{figure}
\centering
    \begin{minipage}{.3\textwidth}
    \centering
        \begin{tikzpicture}[scale=0.7,auto=left]
        	\node[style={circle,inner sep=1pt, minimum width=4pt,draw,fill=black!100},label=left:$a$] (a) at  (-4,0) {};
        	\node[style={circle,inner sep=1pt, minimum width=4pt,draw,fill=black!100},label=above:$b$] (b) at  (-3.5,1) {};
        	\node[style={circle,inner sep=1pt, minimum width=4pt,draw,fill=black!100},label=above:$c$] (c) at  (-0.5,1) {};
        	\node[style={circle,inner sep=1pt, minimum width=4pt,draw,fill=black!100},label=right:$d$] (d) at  (0,0) {};
        	\node[style={circle,inner sep=1pt, minimum width=4pt,draw,fill=black!100},label=below:$e$] (e) at  (-0.5,-1) {};
        	\node[style={circle,inner sep=1pt, minimum width=4pt,draw,fill=black!100},label=below:$f$] (f) at  (-3.5,-1) {};
        	\node[style={circle,inner sep=1pt, minimum width=4pt,draw,fill=black!100},label=left:$x$] (x) at  (-3,0) {};
        	\node[style={circle,inner sep=1pt, minimum width=4pt,draw,fill=black!100},label=right:$y$] (y) at  (-1,0) {};
        	
        	\foreach \from/\to in {a/b, f/x, y/e, c/d}
        	\draw (\from) -- (\to);
        	\foreach\from/\to in {b/c, x/y, f/e}
        	\draw[thick,color=blue] (\from) -- (\to);
        	\foreach\from/\to in {a/f,b/x,y/c,e/d}
        	\draw[thick,color=red] (\from) -- (\to);
        \end{tikzpicture}
    \end{minipage}

\vskip 5mm    
    
    \begin{minipage}{.9\textwidth}
    \centering
        \begin{tikzpicture}[scale=0.7,auto=left]

        	\node[style={rectangle, rounded corners, draw},label=below:{},align=center ] (z') at (-5.5,3) {{\color{red}(0,1)\text{-parallel}}};
        	\node[style={rectangle, rounded corners, draw},label=below:{},align=center ] (u') at (-7.7,1.5) {\tiny{(0,1)-para.}};
        	\node[style={rectangle,rounded corners, draw},label=below: {},align=center ] (v') at (-4,1.5) {\tiny{(0,1)-para.}};
	\node[style={rectangle, rounded corners, draw},label=below:{\{a,b\}},align=center ] (w') at (-8.6,-0.5) {\tiny{(1,1)-deg.}};
        	\node[style={rectangle,rounded corners, draw},label=below: {\{c,d\}},align=center ] (t') at (-6.8,-0.5) {\tiny{(1,1)-deg.}};
	\node[style={rectangle, rounded corners, draw},label=below:{\{e,f\}},align=center ] (x') at (-4.8,-0.5) {\tiny{(1,1)-deg.}};
        	\node[style={rectangle,rounded corners, draw},label=below: {\{x,y\}},align=center ] (y') at (-3,-0.5) {\tiny{(1,1)-deg.}};

        	\foreach \from/\to in {z'/u', z'/v', u'/w', u'/t', v'/x',v'/y'}
        	\draw (\from) -- (\to);
        \end{tikzpicture}
        \hskip 5mm
        \begin{tikzpicture}[scale=0.7,auto=left]
        	\node[style={rectangle, rounded corners, draw},label=below:{},align=center ] (z) at (1,3) {{\color{blue}(0,1)\text{-parallel}}};
        	\node[style={rectangle, rounded corners, draw},label=below:{},align=center ] (u) at (-0.4,1.5) {\tiny{(0,1)-para.}};
        	\node[style={rectangle,rounded corners, draw},label=below: {},align=center ] (v) at (3,1.5) {\tiny{(0,1)-para.}};
	\node[style={rectangle, rounded corners, draw},label=below:{\{a,b\}},align=center ] (w) at (0.6,-0.5) {\tiny{(1,1)-deg.}};
        	\node[style={rectangle,rounded corners, draw},label=below: {\{f,x\}},align=center ] (t) at (-1.2,-0.5) {\tiny{(1,1)-deg.}};
	\node[style={rectangle, rounded corners, draw},label=below:{\{c,y\}},align=center ] (x) at (2.5,-0.5) {\tiny{(1,1)-deg.}};
        	\node[style={rectangle,rounded corners, draw},label=below: {\{e,d\}},align=center ] (y) at (4.3,-0.5) {\tiny{(1,1)-deg.}};
        	\foreach \from/\to in {z/u, z/v,u/w,u/t,v/x,v/y}
        	\draw (\from) -- (\to);

        \end{tikzpicture}

    \end{minipage}
    \caption{
        The graph $G$ on the top is a $(1,1)$-cograph, with two different $(1,1)$-cotrees. 
        The internal nodes have the same labels but the partitions of $V$ induced by the leaves of the $(1,1)$-cotrees are not the same.
    }
    \label{two-diff-decomps}
\end{figure}
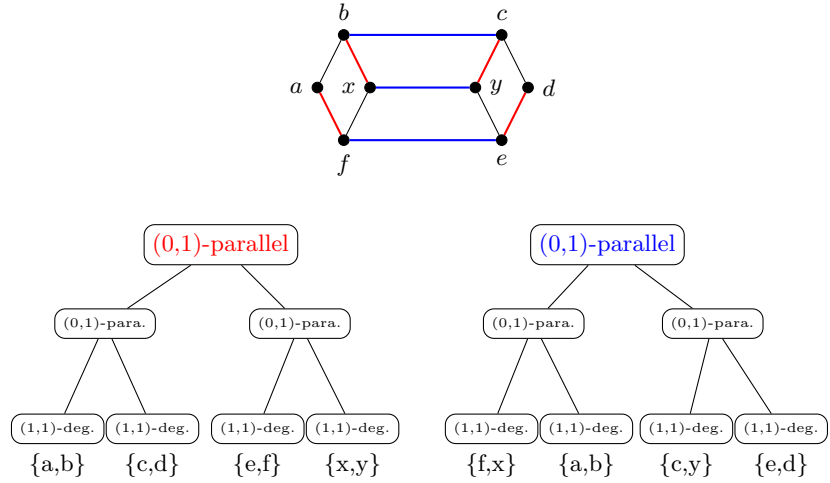

As we shall see next, for $\alpha=0$ and $\beta=1$, the problem of recognizing if a graph admits a $(0, 1)$-parallel decomposition is related to a nice combinatorial problem first studied in~\cite{Graham70}. 
The problem of finding such a decomposition is equivalent to finding a \emph{matching cut set} in a graph,
i.e., an edge cut which is a matching, a well-known problem studied in~\cite{Bonsma89,Chvatal84,Moshi89}. 

Unfortunately however, it turns out --as one might expect-- that finding such a matching cut set in an arbitrary graph is an NP-complete problem, as shown by Chv\'atal in~\cite{Chvatal84}.
\begin{theorem}[\cite{Chvatal84}]
\label{thm:matching_cut set}
    Deciding if a graph has a matching cut set is NP-complete.
\end{theorem}

\ERIC{
\begin{definition}
Let $G=(V,E)$ be a graph and $\mathcal{P} = \{V_1, \dots,V_k\}$ be the partition of $V$
associated with an $(\alpha, \beta)$-modular decomposition of $G$. If every union of parts from $\{V_1, \dots,V_k\}$ is an $(\alpha, \beta)$-module, 
we say that such a decomposition is an \textbf{\boldmath$(\alpha, \beta)$-modular brittle decomposition} of $G$.   
\end{definition}
}

From Theorem~\ref{thm:matching_cut set}, we get the following theorem.

\begin{theorem}
\label{thm:0_1_NP_hard}
For a graph $G=(V,E)$ with $|V| \geq 4$,  deciding if $G$ admits a $(0, 1)$-parallel  (resp.  a $(1,0)$-series) brittle decomposition is NP-complete.
\end{theorem}

\begin{proof}
    Suppose $G$ admits a $(0,1)$-parallel brittle decomposition with partition $\mathcal{P} = \{V_1, \dots,V_k\}$. 
  Since it is a brittle decomposition, $\bigcup\limits_{1<i \leq k}V_i$ is  a $(0,1)$-module,
    and therefore the partition $\{V_1, \bigcup\limits_{1<i \leq k}V_i\}$ is a matching cut set of $G$.
    Using Theorem~\ref{thm:matching_cut set}, deciding if a graph admits a $(0,1)$-parallel decomposition is thus NP-complete.
    For the $(1,0)$-series case, since $G$ has a $(1,0)$-series decomposition if and only if $\overline{G}$ has a $(0,1)$-parallel decomposition by Property \ref{SP}, we are done.
\qed
\end{proof}

Theorem \ref{thm:0_1_NP_hard} implies that the characterization of $(\alpha, \beta)$-brittle graphs could be not so easy.
Since  $(\alpha, \beta)$-modular decomposition trees are not unique, let us now introduce minimal ones.

\begin{definition} An $(\alpha, \beta)$-modular decomposition tree of a graph $G$ is {\bf minimal}
if the first partition level has a minimal number of parts, among all possible $(\alpha, \beta)$-modular decomposition trees of $G$.
\end{definition}

We can consider two decision problems, depending on whether $\alpha$ and $\beta$ are part of the input or not.

\medskip
\noindent
{\sc Minimal modular decomposition}\\
{\it Input:} Two positive integers $\alpha,\beta$ and a graph $G$.\\
{\it Question:} Does $G$ admit a minimal  $(\alpha,\beta)$-modular decomposition?

\medskip
\noindent
{\sc Minimal $(\alpha,\beta)$-modular decomposition}\\
{\it Input:} A graph $G$.\\
{\it Question:} Does $G$ admit a minimal $(\alpha,\beta)$-modular decomposition?

\medskip

Let us prove the NP-hardness of {\sc Minimal modular decomposition}.

\begin{proposition}\label{thm:0_2_NP_hard}
For every pair of integers $(\alpha, \beta)\neq (0,0)$,
and every graph $G$ with $|V(G)| >4(\alpha + \beta)$,  $G$ cannot admit an $(\alpha,\beta)$-series decomposition and a $(\alpha, \beta)$-parallel decomposition both in two parts, say $\{A_1, A_2\}$ and $\{B_1, B_2\}$, respectively.
\end{proposition}

\begin{proof}
First we notice that any two parts from different partitions must overlap. Indeed, suppose for example 
that $A_2 \subseteq B_2$. We then necessarily have $B_1 \subseteq A_1$. 
We then get that a vertex $x \in A_2$ must be $\alpha$-connected and $\beta$-non-connected to $B_1$, a contradiction.

Since $V(G)= \cup_{i,j\in\{1,2\}} A_i \cap B_j$ and $|V(G)| > 4(\alpha + \beta)$, one of these four sets, say $C=A_1 \cap B_1$, has at least $\alpha + \beta +1$ elements. 
But then, every $x \in A_2\cap B_2$ is $\alpha$-connected  and $\beta$-non-connected to $C$, a contradiction.\qed 
\end{proof}

Therefore if $G$ admits a minimal $(\alpha, \beta)$-modular decomposition with only two parts, then it is either an $(\alpha,\beta)$-series decomposition or an $(\alpha, \beta)$-parallel decomposition, but not both.

\begin{corollary}\label{thm:0_2_NP_hard}
Minimal modular decomposition is NP-complete.
\end{corollary}

\begin{proof}
Since computing  a minimal $(0,1)$-modular decomposition tree is enough  to conclude if $G$ has matching cut set or not. But by Theorem \ref{thm:matching_cut set} such a computation is NP-complete. So for $\alpha=0$ and $\beta=1$ the problem is NP-complete.
\qed \end{proof}

We also think that Minimal $(\alpha, \beta)$-modular decomposition is NP-complete for every $(\alpha, \beta)\neq (0,0)$.
In fact our  result \ref{thm:0_2_NP_hard} is not enough yet to prove that simply computing an $(\alpha, \beta)$-modular decomposition tree is NP-hard, which we believe to be true, since, as shown in Figure~\ref{two-diff-decomps}, a given graph may have two different $(0, 1)$-parallel decompositions. 
Furthermore it could be the case that one is brittle and the other not. 
We leave this question as a conjecture.

\begin{mdframed}[style=MyFrame]
\begin{conj}
For every $(\alpha, \beta)$ with $\max\{\alpha,\beta\}>1$,  
finding an $(\alpha, \beta)$-modular decomposition of a graph $G$ is an NP-complete problem. 
\end{conj}
\end{mdframed}

In the above-mentioned work~\cite{Chvatal84}, Chv\'atal showed that the matching cut set problem is NP-complete on graphs with maximum degree four, and polynomial on graphs with maximum degree three. 
In fact the problem of finding a perfect matching cut set is also NP-hard~\cite{HeggernesT98}.
 
On the other hand, computing a matching cut set in the following graph classes is polynomial: 
\begin{itemize}
    \item graphs with maximum degree three~\cite{Chvatal84},
    \item weakly chordal graphs and line-graphs~\cite{Moshi89},
    \item Series-Parallel graphs~\cite{PatrignaniP01},
    \item claw-free graphs and graphs with bounded clique-width, as well as graphs with bounded treewidth~\cite{Bonsma89},
    \item graphs with diameter 2~\cite{BorowieckiJ08}, and 
    \item ($K_{1,4}, K_{1,4}+e$)-free graphs~\cite{KratschL16}.
\end{itemize}

Therefore, to check whether a graph in any of the above classes is a $(1,1)$-cograph, it suffices to run the corresponding matching cut set algorithm, either on the graph itself or on its complement. 

\medskip

From the practical side, a particular subclass of $(0,1)$-parallel graphs has been introduced and studied in network theory~\cite{LTTH2004}, \ERIC{namely the class of networks obtained by starting from the one vertex graph and, at step $t$, taking two graphs obtained in $t-1$ steps and joining them by a perfect matching}. 
This class, called \textbf{Matching Composition Networks} in~\cite{WLCFF19}, contains all hypercubes, as well as all crossed, twisted and M\"{o}bius hypercubes.  In general, $PMG(k)$ is a family of graphs recursively defined, that starts with all connected graphs on $k$ vertices and, at every step, add any graph that can be obtained by selecting two graphs within the family having the same order and joining them with a perfect matching. 

More formally, the family $\PMG(4)$, for instance, is defined as follows.

\begin{definition}[$\PMG(4)$]
\label{def:PMG}
    We start with the following seven connected graphs on four vertices: 
    \begin{itemize}
        \item $P_4, C_4, K_4, K_{1,3}$, 
        \item a triangle with a pending edge,
        \item two triangles having an edge in common.
    \end{itemize}
    At every step, a new graph is obtained from two graphs in the family having the same order by joining them with a perfect matching.
\end{definition}

Knowing that $\PMG(1) = \PMG(2)$ and that they contain hypercubes, crossed, twisted and M\"{o}bius hypercubes, we end this section with the following recognition problem.

\begin{mdframed}[style=MyFrame]
\begin{openproblem}  
Given a graph $G=(V,E)$ with $|V| = 2^n$, what is the complexity of recognizing whether $G$ belongs to $\PMG(k)$ or to one of its non-trivial subclasses?
\end{openproblem} 
\end{mdframed}

\subsection{The Structure of $(\alpha, \beta)$-Prime Graphs}

\begin{proposition}
\label{prop:5_prime}
The only $(1,1)$-prime graph of order~$5$ is $C_5$.
\end{proposition}

\begin{proof}
Let $C_5 = [a,b,c,d,e]$. 
To prove that $C_5$ is a $(1,1)$-prime graph, we just have to prove that every subset of four vertices is not a $(1,1)$-module, which is obvious since for every such subset $A$, the remaining vertex not in $A$ is connected by exactly two edges to $A$. 
A systematic study of all the other graphs of order~$5$ (including the bull) shows that they all have a non-trivial $(1,1)$-module.
\qed
\end{proof}

Notice here that the Petersen graph can be obtained by a $(0,1)$-parallel operation made on two copies of a $C_5$.

Obviously, we have the following inclusion: 
For all $\alpha \leq \alpha'$ and all $\beta \leq \beta'$, the family of $(\alpha, \beta)$-prime graphs is included in the family of $(\alpha', \beta')$-prime graphs. 
But can we improve this result?
In the standard setting, the prime graphs are nested. 
In particular, $P_4$ is the smallest prime graph, and all primes on $n$ vertices contain a prime subgraph on either $n-1$ or $n-2$ vertices, as shown in~\cite{ST91}.

We pose the following problem.

\begin{mdframed}[style=MyFrame]
\begin{openproblem} 
Are the $(\alpha, \beta)$-prime graphs nested? 
\end{openproblem} 
\end{mdframed}

\section{Computing the minimal $(\alpha, \beta)$-modules} 
\label{sec:algorithms}

Despite the negative hardness results in the previous sections, we shall now examine how to compute all minimal $(\alpha, \beta)$-modules of a given graph in polynomial time.
As mentioned earlier, non-trivial $(\alpha, \beta)$-modules have strictly more than $\alpha + \beta + 2$ elements; and since they are closed under intersection, $(\alpha, \beta)$-modules have an underlying graph convexity, and thus (see Algorithm~\ref{alg:naive}), we can compute the minimal $(\alpha, \beta)$-module $M(A)$ that contains a given set $A$ with $|A| > \alpha + \beta + 2$, by computing a \textbf{modular closure} via $(\alpha, \beta)$-splitters. 
In fact, we build a series of subsets $M_i$ that starts with $M_0 = A$ and satisfies $M_i \subseteq M_{i+1}$ for every $i\geq 0$.

\medskip

\begin{algorithm}[t]
\label{alg:naive}
    \LinesNumbered
    \SetAlgoVlined
    \SetKwInput{Input}{Input} 
    \SetKwInput{Output}{Output}
    \caption{Computing minimal $(\alpha, \beta)$-modules.}
    \label{naiveminimal}
    \Input{A graph $G = (V, E)$ and a set $A \subseteq V$ with $|A| \geq \alpha +\beta +2$.}
    \Output{$M(A)$, the unique minimal $(\alpha, \beta)$-module that contains $A$}
        $M_0 \leftarrow A$, $i \leftarrow 0$ \;
        $S \leftarrow \{x \in V\setminus M_0$  :  $\beta <|N(x)\cap M_0|< |M_0|$$-\alpha \}$\;
        \While{$S \neq \emptyset$}{
        	$i \leftarrow i+1$\;
            $M_i \leftarrow M_{i-1} \cup S$\;
            $S\leftarrow \{x \in V\setminus M_i$  :   $\beta <|N(x)\cap M_i|< |M_i|$$-\alpha \}$\;
        }
        $M(A) \leftarrow M_i$\; 
\end{algorithm}

\begin{proposition} 
    Algorithm~\ref{alg:naive} computes the unique minimal $(\alpha, \beta)$-module that contains $A$
    \ERIC{in $O(m \cdot n)$} time.
\end{proposition}

\begin{proof} 
    If $A$ is an $(\alpha, \beta)$-module, then in line 2, $S=\emptyset$; otherwise, all the elements of $S$ have to be added into $M(A)$. In other words, using Lemma~\ref{lem:splitf}, there is no $(\alpha, \beta)$-module $M$ such that : $A \subsetneq M \subsetneq A \cup S$.
    At the end of the while loop, either $M_i=V$ or we have found a non-trivial $(\alpha, \beta)$-module that contains $A$.
    
    \ERIC{This algorithm obviously runs in $O(m \cdot n)$ time.}
    \qed
\end{proof} 

Algorithm~\ref{alg:less_naive} proposes a different implementation that uses a graph search approach to compute the minimal $(\alpha, \beta)$-module containing $A$. This will allow us to achieve a linear running time.

\begin{algorithm}[t]
\label{alg:less_naive}
\LinesNumbered
\SetAlgoVlined
\SetKwInput{Input}{Input} \SetKwInput{Output}{Output}
\caption{Computing minimal $(\alpha, \beta)$-modules.}
    \Input{A graph $G$ and $A \subseteq V(G)$ with $|A| \geq \alpha +\beta +2$.}
    \Output{$M(A)$, the minimal $(\alpha, \beta)$-module that contains $A$}
    
    $OPEN \leftarrow A$\; $M(A) \leftarrow \emptyset$\; 
    \ForEach{$u \in V$}{
        $CLOSED(u) \leftarrow FALSE $; $edge(u) \leftarrow 0$; $non$-$edge(u) \leftarrow 0$\;
    }
    \While{$OPEN \neq \emptyset$}{
        Select a vertex $z$ from $OPEN$ and delete $z$ from $OPEN$\;
		Add $z$ to $M(A)$\;
		$CLOSED(z) \leftarrow TRUE$\;
        \ForEach{$u$ neighbour of $z$}{
            \If{$CLOSED(u) = FALSE$ and $u \notin M(A)$}{
                $edge(u) \leftarrow edge(u)+1$\;
                \If { $\beta < edge(u)$ and $\alpha < non$-$edge(u)$ }{ 
                    Add $u$ to OPEN
                }
            }
         }
        \ForEach{$v$ non-neighbour of $z$}{
            \If{$CLOSED(u) = FALSE$ and $u \notin M(A)$}{
                $non$-$edge(v) \leftarrow non$-$edge(v)+1$\;
                \If { $\beta < edge(u)$ and $\alpha < non$-$edge(u)$ }{ 
                Add $v$ to OPEN
                }
             }
        }
    } 
\end{algorithm}

\begin{theorem}
\label{thm:linear_time}
    Algorithm~\ref{alg:less_naive} can be implemented in $O(m+n)$ time.
\end{theorem}

\begin{proof} 
    We can implement Algorithm~\ref{alg:less_naive} as a kind of a graph search, using an algorithm less naive than Algorithm~\ref{alg:naive}. 
    Algorithm~\ref{alg:less_naive} also computes the minimal $(\alpha, \beta)$-module that contains $A$ in a graph search manner. 

    At the end of Algorithm~\ref{alg:less_naive}, the set $M(A)$ contains a minimal $(\alpha, \beta)$-module that contains $A$.
    At first glance, this algorithm requires $O(n^2)$ operations, since for each vertex we must consider all its neighbours and all its non-neighbours.
    
    However,  if we use a partition refinement technique as defined in~\cite{HabibPV99}, starting with a partition of the vertices as $\mathcal{P} = \{A, V\setminus A\}$. 
    We then keep in the same part, $B(i, j)$, vertices $x, y$ with $edge(x)=edge(y)=i$ and $non$-$edge(x)=non$-$edge(y)=j$.
    This way, when visiting a vertex $z$, it suffices to compute 
    \begin{align*}
        B'(i+1, j) &= B(i, j) \cap N(z) \text{, and }\\
        B''(i, j+1) &= B(i, j)-N(z),
    \end{align*} 
    for each part $B(i, j)$ in the current partition. 
    This can be done in $O(|N(z)|)$ time. 
    
    It should be noted that the parts need not to be sorted in the current partition, and we may have different parts with the same (edge, non-edge) values.
    
    Algorithm~\ref{alg:less_naive} can thus be implemented in $O(m+n)$ time.
    \qed 
\end{proof}

\begin{theorem}
\label{pasimal}
    Using Algorithm~\ref{alg:less_naive}, one can compute all the minimal non-trivial $(\alpha, \beta)$-modules of a given graph in $O(m \cdot n^{\alpha+\beta +1})$ time.
\end{theorem}

\begin{proof}
To do so, it suffices to use Algorithm~\ref{alg:less_naive} starting from every subset of $\alpha + \beta + 2$ vertices.
    There exists $O(n^{\alpha + \beta +2})$ such subsets. 
    This will therefore give us a straight $O(n^{\alpha + \beta +2})$ running time. 
    However, we can use partition refinement to out advantage by using the neighbourhood of one vertex exactly once. 
    But a vertex can belong to at most $n^{\alpha + \beta +1}$ parts in the partition refinement, which yields an algorithm with $O(m \cdot n^{\alpha + \beta +1})$ running time. 
    \qed 
\end{proof}

\noindent
\textbf{Remark}: If we consider the $(0,0)$ case, i.e., the standard case, in Algorithm~\ref{alg:less_naive}, we find the implementation of the algorithm in~\cite{JSC72}, which also computes all the minimal modules in $O(m\cdot n)$ time -- to be compared to the original one in $O(n^4)$ time.

\begin{corollary}
\label{cor:overlap}
    Using theorem~\ref{pasimal}, one can compute a covering of $V$ with an overlapping family of minimal $(\alpha, \beta)$-modules in $O(m \cdot n^{\alpha+\beta +2})$ time. Moreover, the overlapping of any two members of the obtained covering is bounded by ${\alpha+\beta} +1$.
\end{corollary}

\begin{proof}
    Using theorem~\ref{pasimal}, we can compute an overlapping family of minimal $(\alpha, \beta)$-modules in $O(m \cdot n^{\alpha+\beta +1})$ time. But this family can possibly not be a full covering of $V$ since some vertices may not belong to any minimal non-trivial $(\alpha, \beta)$-module. 
    To obtain a full covering, we then simply add the remaining vertices as singletons.
    \qed 
\end{proof}

Corollary~\ref{cor:overlap} can be very interesting if we are looking for overlapping communities in social networks, where the overlapping is bounded by $\alpha + \beta + 1$.

Going a step further, we can use Theorem~\ref{thm:presquepar} and merge every pair $A, B$ of $(\alpha, \beta)$-modules with $|A \cap B | \geq \alpha + \beta +1$, either by keeping $A \cup B$ as a $(2\alpha, 2 \beta)$-module, or by computing $M(A \cup B)$, the minimal $(\alpha, \beta)$-module that contains $A \cup B$.
This depends however on the structure of the maximal $(\alpha, \beta)$-modules, and unfortunately we do not know yet under which conditions there exists a unique partition into maximal $(\alpha, \beta)$-modules.

\begin{corollary}
\label{cor:prime_check}
    Checking if a graph is $(\alpha, \beta)$-prime can be done in $O(m \cdot n^{\alpha+\beta +1})$ time.
\end{corollary}

\begin{proof}
    Easy using Theorem~\ref{pasimal}.
    \qed 
\end{proof}

\section{An Application on Bipartite Graphs}
\label{sec:applications}

In this section, let $G = (X, Y, E)$ be a bipartite graph with parts $X$ and $Y$. 
By allowing $\alpha+\beta$ errors in the decomposition, $(\alpha, \beta)$-modules can be made up with vertices from both $X$ and $Y$.
However, in some applications, we are forced to consider $X$ and  $Y$ separately.
Consider for instance the setting where $X$ and $Y$ represent the sets of customers and products, or the sets of DNA sequences and organisms, in which case one would want to find regularities on each side of the bipartition. 

\begin{definition}
\label{def:F_alpha_beta}
For a given bipartite graph $G = (X, Y, E)$, we let 
  \begin{align*}
      {\cal F}_{\alpha, \beta}(X) = \{ M : M \text{ is an } (\alpha, \beta)\text{-module of $G$ and } M \subseteq X \}.
  \end{align*}
  Note that $X$ is not always an $(\alpha, \beta)$-module of $G$.
\end{definition}

\begin{proposition}
\label{prop:F_almost_partitive}
 For every two sets  $A, B \in {\cal F}_{\alpha, \beta}(X)$,
  $A \cap B$, $A \setminus B$ and $B\setminus A$ are all in ${\cal F}_{\alpha, \beta}(X)$.
\end{proposition}

\begin{proof}
    Using Theorem~\ref{thm:prepar}, the only $(\alpha, \beta)$-splitters of the sets $A \setminus B$ and $B\setminus A$ must belong to $A \cap B$; but since $A, B \subseteq X$, and $X$ is an independent set, this is not possible.
    \qed
\end{proof}

It should be noticed here that for $A \subseteq X$, the minimal $(\alpha, \beta)$-module that contains $A$ does not always belong to ${\cal F}_{\alpha, \beta}(X)$, since we may have to add $(\alpha, \beta)$-splitters from $Y$. Therefore we have to use an algorithmic approach different from those developed in the previous section in order to compute ${\cal F}_{\alpha, \beta}(X)$.

\begin{definition}
\label{def:etwins}
Two sets  $A,B \subseteq V$, $A\neq B$, with $|A|=|B|=\alpha + \beta+1$, are said to be  \textbf{false \boldmath$(\alpha, \beta)$-twin}
(resp. \textbf{true \boldmath$(\alpha, \beta)$-twin}) in $G$ if they satisfy the following three conditions:
    \begin{enumerate}
        \item $A \cup B$ is an $(\alpha, \beta)$-module,
        \item $\forall x \in A$, $ x \in N_{\alpha}(B)$ 
        (resp. $ x \in \overline{N}_{\beta}(B)$),
                \item $\forall y \in B$, $ y \in N_{\alpha}(A)$ 
        (resp. $ y \in \overline{N}_{\beta}(A)$).
    \end{enumerate}
\end{definition}

Observe that $A$ and $B$ are false $(\alpha, \beta)$-twin sets in $G$ if and only if $A$ and $B$ are true $(\alpha, \beta)$-twin sets in $\overline{G}$.

\begin{proposition}
\label{prop:equivalence}
    Being (true or false) $(\alpha, \beta)$-twin is an equivalence relation on subsets of vertices.
\end{proposition}

When applying Definition~\ref{def:etwins} to bipartite graphs, we obviously only have false $(\alpha, \beta)$-twin sets.

\begin{proposition}
\label{prop:union_false_twins}
    A set $M \subseteq X$ is an $(\alpha, \beta)$-module if and only if $M$ is a union of false $(\alpha, \beta)$-twin sets.
\end{proposition}

\begin{proof}
    Let $A, B \subseteq M$, with $|A|=|B|=\alpha +\beta +1$. 
    Pick any vertex $z \in Y$.
    If $z \in N_{\alpha}(A)$, then $z \in N_{\alpha}(M)$ and therefore $z \in N_{\alpha}(B)$. 
    Therefore, $A$ and $B$ are false $(\alpha, \beta)$-twin sets since they are both included in $X$.
    
    The converse directly follows from Definition~\ref{def:etwins}.
    \qed 
\end{proof}

Consequently, in terms of $(\alpha + \beta +1)$-tuples, the sets of false $(\alpha, \beta)$-twin sets partition the $(\alpha + \beta +1)$-tuples.
Furthermore, using the notion of false $(\alpha, \beta)$-twin sets, we obtain the following theorem (recall that for a graph $G$, ${\cal F}_{\alpha, \beta}$ is the set of its $(\alpha, \beta)$-modules whose elements are in $X$). 

\begin{theorem}
\label{thm:maximal_bipartite}
    For a given bipartite graph $G=(X, Y, E)$, the maximal elements of ${\cal F}_{\alpha, \beta}(X)$ can be computed in $O(n^{\alpha + \beta}(n + m))$ time.
\end{theorem}

\begin{proof}
To do so, we first build an auxiliary bipartite graph, $G'=( {\cal A}, Y, E(G'))$,
which represents the labelled incidence graph of the $(\alpha + \beta +1)$-tuples of vertices of $X$. 
The set of vertices of $G'$ is  thus the set 
$\cal A$  of these $(\alpha + \beta +1)$-tuples,

By Lemma \ref{lem:comptages}.4, we know that every such tuple $T$ yields a partition of $Y$ into $N_{\alpha}(T)$ and $\overline{N}_{\beta}(T)$. The set of edges of $G'$ is then defined by setting,
for every $T\in {\cal A}$ and $y \in Y$,
$$Ty \in E(G') \text { if and only if } y \in N_{\alpha}(T),$$ 
which implies
$$Ty \notin E(G') \text{ if and only if } y \in \overline{N}_{\beta}(T).$$
Since every vertex in $Y$ belongs to at most $O(n^{\alpha + \beta})$ tuples from $\cal A$, the number of edges in $E(G')$ is in $O(m\cdot n^{\alpha + \beta})$. 

Given the auxiliary graph $G'$, we now partition $\cal A$ into false twins. 
To this aim, we use every vertex in $Y$ to refine $\cal A$ with respect to $(\alpha, \beta)$-neighbourhood. 
This can be done in $O(n^{\alpha + \beta +1}+n^{\alpha + \beta}\cdot m)$ time, using standard partition refinement techniques~\cite{HabibPV99}. 
  
Let $Q = \{{\cal A}_1, \dots {\cal A}_k\}$ be such a partition. 
We prove the following claim.

\medskip
\noindent
\textbf{Claim:} {\it No element of ${\cal F}_{\alpha, \beta}$ can contain two $(\alpha + \beta +1)$-tuples from different parts of $Q$.}

\begin{proof}
Let $A_i \in {\cal A}_i$, $A_j \in {\cal A}_j$, with $i \neq j$,
and $S$ be a subset of $X$ such that $A_i \cup A_j \subseteq S$.
Since $i \neq j$, there is a vertex $y \in Y$ such that (w.l.o.g.)  $A_i \in N_{\alpha}(y)$ and $A_j \in \overline{N}_{\beta}(y)$. 
Hence we have
$$ |A_i| -\alpha \leq |S \cap N(y)|  \leq |S| -|A_j| +\beta,$$
which gives
$$\beta +1 \leq |S \cap N(y)| \leq |S| - \alpha -1,$$
and thus $y$ is an $(\alpha, \beta)$-splitter for $S$.
\qed
\end{proof}

Therefore, to find the maximal elements of ${\cal F}_{\alpha, \beta}$, 
we can restrict the search to the ${\cal A}_i$'s. 
Let us now examine how to generate them. To this aim, we define a labelling $\lambda$ that assigns to each ordered pair $(y,A)$, with $y\in Y$ and $A\in {\cal A}$, a subset of $A$ as follows.

\begin{itemize}
\item If $yA \in E(G')$ and $a_1, \dots, a_k$,  $k\leq \alpha$, are the vertices from $A$ non adjacent to $y$, then we set $\lambda(y,A)= \{a_1, \dots, a_k\}$.
\vskip 2ex

\item Symmetrically, if $yA \notin E(G')$ and $a_1, \dots, a_h$, $h\leq \beta$, are the vertices from $A$ adjacent to $y$, then we set $\lambda(y,A)=\{a_1, \dots, a_h\}$.
\end{itemize}
This labelling can be done while constructing the graph $G'$.

Then, a maximal element $F$ of ${\cal F}_{\alpha, \beta}$ is just a maximal  union of elements of some ${\cal A}_i$, $1\le i\le k$, satisfying the following:
 
\begin{itemize}
\item For every vertex $y\in Y$,
\begin{itemize}
\item if every element of ${\cal A}_i$ is adjacent to $y$,
   then  $|\cup_{A \in F} \lambda(y,A) | \leq \alpha$, 
\item otherwise, $|\cup_{A \in F} \lambda(y,A) | \leq \beta$. 
\end{itemize}
\end{itemize} 

Note that all vertices in ${\cal A}_i$ are false twins, since the graph $G'$ is bipartite, and therefore connected the same way to $Y$.

To produce these maximal sets, we start with $\alpha=\beta=0$, in which case the only maximal module has an empty label. Let $M_0$ denote this module and ${\cal M}_{0,0}=\{M_0\}$ denote the set of maximal elements at this step.
We then increase either $\alpha$ or $\beta$ by one, and recursively compute the new set, ${\cal M}_{\alpha+1,\beta}$ or ${\cal M}_{\alpha,\beta+1}$, of maximal elements from the previously computed set ${\cal M}_{\alpha,\beta}$
(note that every maximal $(\alpha, \beta)$-module is contained in a maximal $(\alpha+1, \beta)$-module and in a maximal $(\alpha, \beta+1)$-module as well). 

For $\alpha=\beta=0$, $M_0$ is unique. 
For $\alpha=\beta=1$, there are at most $|Y|^{2}$ maximal  $(1, 1)$-modules in ${\cal F}_{1,1}$.
Hence, there are at  most  $|Y|^{\alpha +\beta}$  maximal $(\alpha, \beta)$-modules in   ${\cal F}_{\alpha, \beta}$. 
 This computation is therefore bounded in the whole by
 $$(\alpha+ \beta) (\Sigma_{i=0}^{i=\alpha+\beta} |Y|^{i}) \cdot ( |X|^{\alpha + \beta +1}),$$
which is in the order of  
$O((|Y|^{\alpha +\beta+1}) \cdot ( |X|^{\alpha + \beta +1}))$.
\qed 
\end{proof}

Note that these maximal elements of ${\cal F}_{\alpha, \beta}(X)$ may overlap.
It remains to test the quality of the covering obtained on some real data graphs. 
We leave this as something to explore for data analysts.

\section{Conclusion}
\label{sec:conclusion}
Before we conclude, we want first to expose the reader to a different way to approach the approximation of modules.

\subsection{$k$-splitter Modules: An Alternate Approximation}

Another natural way to approach the problem of approximating modules is by restricting the number of splitters a module can have. 
Recall that in the standard modular decomposition setting, a splitter of a module $M$ in a graph $G=(V,E)$  is a vertex $v \in V \setminus M$ such that there exists at least two vertices $a,b \in M$ with $av \in E$ and $bv \notin E$. 
By restricting the number of splitters outside a module, we get the following definition -- which intuitively just allows at most $k$ ``errors'' in terms of connectivity. 

\begin{definition}
    For a given graph $G=(V,E)$, a subset $M$ of $V$ is a \textbf{\boldmath$k$-splitter module} if 
    $M$ has at most $k$ splitters.
\end{definition}

Notice then that by setting $k=0$ in the above definition, we recover the standard modular decomposition setting~\cite{HabibP10}, i.e., for every $x \in V \setminus M$,
either $M \cap N(x)=\emptyset$ or $M \cap N(x)=M$. So, for this approximate setting, we will necessarily only consider the case $k < |V(G)|-1$.

We begin with some obvious remarks.

\begin{proposition} If $M$ is a $k$-splitter for $G$, then the following holds.
    \begin{enumerate}
    \item $M$ is a $k'$-splitter module for $G$, for every $k'\geq k$.
    \item $M$ is a $k$-splitter module for $\overline{G}$.
    \item If $s$ is a splitter for $M$, then $s$ is also a splitter for every set $M' \supseteq M$ with $s \notin M'$.
    \end{enumerate}
\end{proposition} 

\begin{proposition}
    The family of $k$-splitter modules of a graph $G=(V,E)$ satisfies the following.
    \begin{enumerate}
        \item Every set $A \subseteq V$ with $|A| \leq 1$ or $|A| \geq |V|-k$ is a $k$-splitter module of $G$. 
        (We call such a set $A$ a \textbf{trivial \boldmath$k$-splitter module}.)
        \item For every two $k$-splitter modules $A,B \subseteq V$  of $G$ with $A\cap B\neq \emptyset$, 
        $A \cup B$ is a $2k$-splitter module of $G$.
        \item For every two $k$-splitter modules $A,B \subseteq V$ of $G$  with $A\cap B\neq \emptyset$, 
        $A \cap B$ is a $2k$-splitter module of $G$.
    \end{enumerate}
\end{proposition}

\begin{proof}\mbox{}
    \begin{enumerate}
        \item This follows from the definition.
        \item There cannot be a splitter of $A\cup B$ that is not a splitter of either $A$ or $B$ since $A\cap B\neq \emptyset$. 
        We get the $2k$ in the worst case, when both $A$ and $B$ have two disjoint sets of $k$ splitters outside of $A\cup B$.
        \item A splitter of $A\cap B$ in $V \setminus (A\cup B)$ is a splitter of both $A$ and $B$.
        A splitter of $A\cap B$ in $A$ is a splitter of $B$ and a splitter of $A\cap B$ in $B$ is a splitter of $A$. 
Therefore, the number of splitters of $A\cap B$ is at most the sum of the numbers of splitters of $A$ and $B$, i.e., $2k$.
    \end{enumerate}
    \qed 
\end{proof}

\begin{proposition}
    If $A$ and $B$ are two non-trivial $k$-splitter modules of a graph $G=(V,E)$, then 
    $A \setminus B$ is a $(k+|A\cap B|)$-module of $G$.
\end{proposition}

\begin{proof}
    There are at most $k$ splitters of $A \setminus B$ in $V \setminus A$, and at most $|A\cap B|$ splitters of $A \setminus B$  in $A\cap B$.
    \qed 
\end{proof}

\begin{proposition} 
For a graph $G=(V,E)$ and a subset $A \subseteq V$,
      there may exist different  minimal (under inclusion) $k$-splitter modules containing $A$.
\end{proposition}

\begin{proof}
    Suppose $A$ admits $k+1$ splitters. 
    Every one of these $k+1$ splitters can be added  to $A$ in order to obtain a $k$-splitter module.
    \qed 
\end{proof}

In conclusion, this approximation variation is not closed under intersection, unfortunately. There was no way to define some sort of convexity, and thus no easy way to define a closure operator with this notion, which is why we have focused our study on $(\alpha, \beta)$-modules instead.

\subsection{Conclusions and Perspectives}

In this work, we introduce a new notion of modular decomposition relaxation. 
This notion of $(\alpha, \beta)$-module yields many interesting questions, both 
from a theoretical or practical point of view. 
Standard modular decomposition is too restrictive for graphs that arise from real data; 
do $(\alpha, \beta)$-modules indeed often arise in this setting? 
We believe this relaxation of modular decomposition can definitely find applications in practice. 
On the theory side, this new combinatorial decomposition may help to better understand graph structuration that can be obtained when grouping vertices that have similar neighbourhood. Such an idea has been successfully used with the notion of twin-width~\cite{bonnet2020twinwidth,bonnet2021twin}.
 Furthermore it is related to fundamental combinatorial objects as for example matching cutsets and their generalization. Our  work leaves many interesting questions open (five open questions and one conjecture), as the study of $(\alpha, \beta)$-prime graphs for instance. 
We have also exhibited new classes of graphs, such as the $(1,1)$-cographs, that contain many interesting subclasses and on which it would be interesting to consider the Erd\"{o}s-Hajnal conjecture \cite{ErdosH89}, which holds for cographs and is closed under substitution~\cite{BKC19}.
Our work could also be related to the very interesting new generalization of cographs introduced in \cite{bonnet2020twinwidth}.

We have also presented polynomial time algorithms that we believe could all be improved. 
It is important however to keep in mind that since the number of unions of overlapping minimal modules can be exponential, it is thus hard to compute from the minimal $(\alpha, \beta)$-modules some hierarchy of modules.
However, perhaps a better way to decompose a graph is to first compute the families of minimal modules with small values of $\alpha$ and $\beta$, and then consider a hierarchy of overlapping families.

\bibliography{decomp}
\end{document}